\documentclass{article}

\usepackage{amsmath,amsthm,amssymb,bm,color,graphicx,mathrsfs}
\usepackage{amscd,verbatim}
\usepackage[breaklinks=true,colorlinks=true,linkcolor=blue,urlcolor=blue,citecolor=blue]{hyperref}
\usepackage{graphicx}
\usepackage{subfigure}
\usepackage{comment}
\usepackage{lipsum}
\usepackage[arrow,matrix,curve]{xy}
\usepackage{authblk}
\usepackage[numbers]{natbib}
\usepackage{braket}

\usepackage{tikz-cd}
\usetikzlibrary{decorations.markings}
\tikzset{->-/.style={decoration={
  markings,
  mark=at position #1 with {\arrow{>}}},postaction={decorate}}}
\tikzset{->-/.default=0.5}

\usetikzlibrary{hobby}

\newtheorem{theorem}{Theorem}
\newtheorem{definition}{Definition}
\newtheorem{proposition}{Proposition}
\newtheorem{lemma}{Lemma}
\newtheorem{corollary}{Corollary}

\theoremstyle{remark}
\newtheorem{remark}{Remark}
\newtheorem{example}{Example}

\newcommand\wt{\widetilde}
\newcommand\wh{\widehat}

\newcommand{\CC}{{\mathbb C}}
\newcommand{\PP}{{\mathbb P}}
\newcommand{\RR}{{\mathbb R}}
\newcommand{\TT}{{\mathbb T}}
\newcommand{\ZZ}{{\mathbb Z}}
\newcommand{\NN}{{\mathbb N}}

\newcommand{\vect}[1]{\boldsymbol{#1}}

\begin{document}

\title{On spectral flow and Fermi arcs}

\author[1]{Guo Chuan Thiang}
\affil[1]{Beijing International Center for Mathematical Research, Peking University, No.\ 5 Yiheyuan Road, Haidian District, Beijing 100871, China}
\affil[1]{School of Mathematical Sciences, University of Adelaide, SA 5005, Australia}

\renewcommand\Authands{ and }


\maketitle

\begin{abstract}
We introduce spectral flow techniques to explain why the Fermi arcs of Weyl semimetals are topologically protected against boundary condition changes and perturbations. We first analyse the topology of a certain universal space of self-adjoint half-line massive Dirac Hamiltonians, and then exploit its non-trivial and homotopy invariant spectral flow structure by pulling it back to generic Weyl semimetal models. The homological perspective of using Dirac strings/Euler chains as global topological invariants of Weyl semimetals/Fermi arcs, is thereby analytically justified.
\end{abstract}

\section*{Introduction}
In solid-state physics, 3D \emph{Weyl semimetals} \cite{AMV,Witten,MT} are characterised by the presence of topological spectral band crossings in momentum space\footnote{Whereas translations in Euclidean space are denoted by $\RR^d$, the Fourier transform, or momentum space, variable is denoted $\widehat{\RR}^d$.} ($\widehat{\RR}^3$), and the Weyl Hamiltonian models the semimetal near such a crossing point (which is called a \emph{Weyl point}), see \S\ref{sec:band.crossings}. On the surface of a Weyl semimetal, \emph{Fermi arcs} of boundary-localised states have been experimentally found \cite{Xu, Xu2, Lv, Morali, Souma} to connect \emph{pairs} of Weyl points with opposite chiralities/indices (projected to the surface momentum space $\widehat{\RR}^2$). The robustness of the global connectivity of surface spectra is much more surprising than the mere existence of gapless (zero energy) edge states, and should be considered an indispensable part of what it means for Fermi arcs to be \emph{topological}. In this paper, we establish that Fermi arc topology is prefigured by the non-trivial \emph{homotopy invariant spectral flows} associated to the Hamiltonian operator for the semimetal. Conceptually, the upshot is that the homological (``Dirac-stringy'') perspective of Fermi arcs, developed in \cite{MT, MT1} and briefly recalled in \S\ref{sec:homology.perspective}, is analytically justified in a rigorous way.

Our spectral flow analysis begins with that of the (right-handed) \emph{Weyl Hamiltonian} on Euclidean $\RR^3$,
\begin{equation}
H^{\rm Weyl}=-i\nabla\cdot\vect{\sigma}=\begin{pmatrix}-i\partial_z & -i\partial_x-\partial_y \\ -i\partial_x+\partial_y & i\partial_z\end{pmatrix}.\label{eqn:Weyl.Hamiltonian}
\end{equation}
Here $\vect{\sigma}=(\sigma_x,\sigma_y,\sigma_z)$ is shorthand for the vector of Pauli matrices. 
Fourier decomposition of $H^{\rm Weyl}$ in the $x$ and $y$ directions gives the family
\begin{equation*}
\widehat{\RR}^2\ni (p_x,p_y)\mapsto H^{\rm 1D}(p_x,p_y)=\begin{pmatrix} -i\frac{d}{dz} & p_x-ip_y \\ p_x+ip_y & i\frac{d}{dz} \end{pmatrix}
\end{equation*}
of Dirac Hamiltonians on the line $\RR$, with the off-diagonal part serving as a ``mass term''. Notice that \emph{all} possible massive Dirac Hamiltonians on $\RR$ are accounted for in this decomposition. 

When a boundary $z=0$ is introduced, the upper half-space Weyl Hamiltonian, $\wt{H}^{\rm Weyl}$, Fourier decomposes instead into Dirac Hamiltonians on a \emph{half-line}, $\wt{H}^{\rm 1D}(p_x,p_y;\gamma)$, parametrised by $(p_x,p_y)\in\widehat{\RR}^2$. Here, an extra boundary condition parameter $\gamma\in{\rm U}(1)$ is needed to ensure self-adjointness of the half-line Dirac Hamiltonians, therefore also of $\wt{H}^{\rm Weyl}$. In \S\ref{sec:spectral.flow.structure}, we will show that the space $\mathcal{M}$ of \emph{all} possible self-adjoint \emph{half-line} Dirac Hamiltonians is naturally \emph{continuously} parametrised by $(p_x,p_y;\gamma)\in\widehat{\RR}^2\times{\rm U}(1)$, with respect to certain important unbounded operator topologies (Corollary \ref{cor:general.Dirac.loop.sf}). The subspace $\mathcal{M}^\prime$ of massless half-line Dirac Hamiltonians is therefore a \emph{continuous} family of unbounded self-adjoint \emph{Fredholm} operators, and loops within $\mathcal{M}^\prime$ have explicitly computable and non-trivial \emph{spectral flow} across zero energy, see \S\ref{sec:spectral.flow.generalities}. 

In particular, for the Fourier decomposition of $\wt{H}^{\rm Weyl}$ into half-line Dirac Hamiltonians, we find that \emph{independently} of boundary conditions, \emph{any} loop winding around the origin has non-zero spectral flow across $0$, so that a zero-energy eigenstate must occur for some operator in the loop. Consequently, there is a one-dimensional locus in $\widehat{\RR}^2$ for which zero-energy spectra occur --- this is the \emph{basic Fermi arc} emanating from $(0,0)$, see Fig.\ \ref{fig:Dirac.spectral.flow}. Exploiting the homotopy invariance of spectral flow, together with a perturbation analysis, we demonstrate in \S\ref{sec:Weyl.Fermi}, the stability of the Fermi arc against \emph{boundary condition changes} (Prop. \ref{prop:boundary.independence}), \emph{chemical potential shifts}, \emph{confining potentials} (Thm.\ \ref{thm:perturbation.independence}), and \emph{gauge transformations} (Prop.\ \ref{prop:Wahl.twisted.sf}). The effect of these modifications is \emph{geometric}, affecting its shape but not its \emph{existence} or \emph{connectivity} from the origin to infinity. Establishing such stabilities is very important, since it is easy to construct examples of spurious Fermi arc phenomena, see \S\ref{sec:spurious.Fermi}. Furthermore, experimentally found Fermi arcs differ from those computed from idealised models in significant ways, see Remark \ref{rem:fuzzy.arc}.

Returning to Weyl semimetals, let us observe that the above local model ($H^{\rm Weyl}$) near one Weyl point can never simultaneously work for a distinct Weyl point, so we cannot capture the connectivity data of Weyl points/Fermi arcs in such first-order models. In \S\ref{sec:top.analysis}, we analyse a second-order model (\S\ref{sec:generic.model}, \S\ref{sec:shifted.Weyl.points}) which captures generic \emph{pairs} of Weyl points, reducing the local description to a classifying map $\widehat{\RR}^2\rightarrow\mathcal{M}$. The spectra of the operators in $\mathcal{M}$ are pulled back to the model under this map, as is the non-trivial spectral flow structure, leading precisely to the above phenomenon of \emph{topological} Fermi arcs connecting projected Weyl points. We also analyse the global spectral flow structure for tight-binding models in \S\ref{sec:tight.binding} using $K$-theory techniques, obtaining as a bonus, a spectral flow perspective of the gap-filling phenomenon/bulk-boundary correspondence for 2D Chern insulators (cf.\ \cite{Braverman}). 

We stress that the real strength of topological arguments is to justify extrapolation of analysis from simple exactly solvable models. It is a hopeless task to fully solve spectral problems in general (or even write down the actual complicated potential terms in the first place). In spite of this, spectral phenomena captured by topological invariants \emph{are} inherited from the simple model, \emph{once it is established that the relevant space of models is continuous in the control parameters}. The latter continuity ($\sim$ finding a good topology on the ``space of models'') is by no means automatic, especially when unbounded operators are involved. For our study of Fermi arcs, the control parameters include, but are not restricted to, boundary condition choices and perturbations.

On the purely mathematical side, although rigorous studies of spectral flow have been available in some form since \cite{AS, APS, Phillips}, these have usually been formulated for bounded self-adjoint Fredholm operators, or applied to elliptic differential operators on \emph{compact} manifolds with boundary. Only more recently, in \cite{BLP, Wahl} have these ideas been extended directly to unbounded self-adjoint Fredholm operators, although applications still generally exclude \emph{non-compact} manifolds with boundary such as the half-line. Our fairly elementary example of $\mathcal{M}^\prime$ provides motivation for further development of unbounded operator spectral flow (developed, e.g.\ in \cite{BLP,Wahl}), and we anticipate applications to the non-compact manifold setting and general bulk-boundary correspondences in physics \cite{LT}. We have formulated this study with a view towards incorporating in a future work, noncommutative spectral flow \cite{Wahl} and noncommutative geometry techniques, in order to address effects of random potentials \cite{PSB} and geometrically complicated boundaries \cite{Thiang,LT}. We also mention that in \cite{CT}, an analogous analysis of 5D Weyl semimetals has been carried out --- ``higher spectral flow'' as encoded in a ``Fermi gerbe'' protects the Fermi arcs in that case.

\section{Spectrum of half-line Dirac Hamiltonians}
It will be convenient to use polar coordinates $(m,\theta)$ rather than $(p_x,p_y)$ to label the Dirac operators on the Euclidean line $\RR$,
\begin{equation}
H^{1D}(p_x,p_y)\rightsquigarrow H^{1D}(m,\theta)=\begin{pmatrix} -i\frac{d}{dz} & me^{-i\theta} \\ me^{i\theta}& i\frac{d}{dz} \end{pmatrix},\qquad (m,\theta) \leftrightarrow (p_x,p_y)\in\widehat{\RR}^2\label{eqn:1D.Dirac}.
\end{equation}
They admit only one sensible domain of self-adjointness in $L^2(\RR)^{\oplus 2}$ --- the first Sobolev space\footnote{Here $H^1(\RR)$ is standard notation for the space of locally absolutely continuous functions with $L^2$ (weak) derivative. Later, we encounter (co)homology groups, similarly denoted $H^n, H_n$, but it should be clear from context what is meant.} $H^1(\RR)^{\oplus 2}$ --- which is often left implicit.
Note that the $m=0$ case has ill-defined $\theta$, but unambiguously refers to the massless Dirac operator. Using Fourier transform in the remaining $z$-variable, it is clear that the (essential\footnote{Recall that a self-adjoint operator $D$ is \emph{Fredholm} if it has finite-dimensional kernel and cokernel, while its \emph{essential} spectrum comprises those $\lambda\in\RR$ such that $D-\lambda$ fails to be Fredholm.}) spectrum is
\begin{equation*}
\sigma(H^{1D}(m,\theta))=\sigma_{\rm ess}(H^{1D}(m,\theta))=(-\infty,-m]\cup[m,\infty).
\end{equation*}

We wish to define \emph{half-line Dirac operators} $\wt{H}^{\rm 1D}(m,\theta)$ acting on $L^2(\RR_+)^{\oplus 2}$, where $\RR_+:=(0,\infty)$, but as discussed in Appendix \ref{sec:deficiency.indices}, the formal expression Eq.\ \eqref{eqn:1D.Dirac} becomes ambiguous. Namely, we are only assured that $\wt{H}^{\rm 1D}(m,\theta)$ is symmetric (i.e.\ formally self-adjoint) on the initial domain $C_0^\infty(\RR_+)^{\oplus 2}$, and it is necessary to specify boundary conditions at $z=0$ to make it self-adjoint. The self-adjoint boundary conditions turn out to be specified by an angular variable\footnote{We will move freely between $\gamma$ as a real number modulo $2\pi$, and $\gamma$ as a phase, i.e.\ $e^{i\gamma}$.} $\gamma\in [0,2\pi]/_{0\sim 2\pi}\cong{\rm U}(1)\ni e^{i\gamma}$.
\begin{definition}\label{defn:half.line.Dirac.domain}
The self-adjoint half-line Dirac Hamiltonians $\wt{H}^{\rm 1D}(m,\theta;\gamma)$ on $L^2(\RR_+)^{\oplus 2}$ are
\begin{align}
\wt{H}^{1D}(m,\theta;\gamma)&=\begin{pmatrix} -i\frac{d}{dz} & me^{-i\theta} \\ me^{i\theta} & i\frac{d}{dz} \end{pmatrix},\qquad\qquad\qquad\qquad (m,\theta)\in\widehat{\RR}^2,\nonumber\\
{\rm Dom}\big(\wt{H}^{\rm 1D}(m,\theta;\gamma)\big)&=\left\{\psi\in H^1(\RR_+)^{\oplus 2}\;\Big{|}\; \psi(0)\propto\begin{pmatrix}1 \\ e^{i\gamma}\end{pmatrix}\right\},\qquad e^{i\gamma}\in{\rm U}(1).\label{eqn:half.line.Dirac.domain}
\end{align}
We denote by $\mathcal{M}$, the universal three-parameter space
\begin{equation}
\wt{H}^{\rm 1D}(m,\theta;\gamma)\in \mathcal{M}\cong\widehat{\RR}^2\times{\rm U}(1)\label{eqn:half.line.Diracs}
\end{equation}
of self-adjoint half-line Dirac Hamiltonians, and by
\begin{equation*}
\mathcal{M}^\prime=(\widehat{\RR}^2\setminus\{0\})\times {\rm U}(1),
\end{equation*}
the subparameter space of the massive ($m>0$) ones.
\end{definition}
The essential spectrum of $\wt{H}^{\rm 1D}(m,\theta;\gamma)$ remains (see Appendix \ref{sec:half.line.spectrum})
\begin{equation}
\sigma_{\rm ess}(\wt{H}^{\rm 1D}(m,\theta;\gamma))= (-\infty,-m] \cup [m,\infty),\label{eqn:Dirac.essential.spectrum}
\end{equation}
so $\mathcal{M}^\prime$ sits inside the space $\mathcal{F}^{\rm sa}$ of \emph{unbounded self-adjoint Fredholm} operators.

\begin{remark} The boundary condition in Eq.\ \eqref{eqn:half.line.Dirac.domain} is physically a (pseudo-)\emph{spin polarisation condition}, as is apparent by rewriting it as 
\begin{equation*}
\psi(0)=\underbrace{\begin{pmatrix} 0 & e^{-i\gamma} \\ e^{i\gamma} & 0\end{pmatrix}}_{(\cos\gamma) \sigma_x+(\sin\gamma)\sigma_y}\psi(0),\qquad e^{i\gamma}\in {\rm U}(1).
\end{equation*}
\end{remark}

\subsection{Discrete spectrum of half-line Dirac Hamiltonians}\label{sec:massive.halfline.Dirac.basic}

Intuitively, $\sigma_{\rm ess}$ accounts for the part of the spectrum which is robust against boundary conditions and perturbations. In particular, as the boundary condition parameter $\gamma$ is varied, we could get discrete spectrum (finite-multiplicity isolated eigenvalues) appearing in the essential spectrum gap $(-m,m)$ of $\wt{H}^{\rm 1D}(m,\theta;\gamma)$.

One can explicitly verify that for $\theta\in(\gamma,\pi+\gamma)$, the eigenvalue-eigenfunction\footnote{By elliptic regularity, we only need to find smooth eigenfunctions.} system of $\wt{H}^{\rm 1D}(m,\theta;\gamma)$ is 
\begin{align}
\wt{H}^{\rm 1D}(m,\theta;\gamma)\cdot\psi_{m,\theta;\gamma}&=m\cos(\theta-\gamma)\cdot\psi_{m,\theta;\gamma},\label{eqn:1D.evalue}\\
\psi_{m,\theta;\gamma}&\propto\;z\mapsto e^{-zm\sin(\theta-\gamma)}\begin{pmatrix} 1 \\ e^{i\gamma} \end{pmatrix}.\label{eqn:1D.efunction}
\end{align}
{\bf Indications of spectral flow.} These eigenfunctions $\psi_{m,\theta;\gamma}$ are ``edge states'', since they are exponentially localised near $z=0$. When $m>0$ and $\gamma$ are fixed while $\theta$ is increased in the interval $(\gamma,\pi+\gamma)$, an eigenvalue emerges out of the bottom of the upper essential spectrum, flows downwards, and then merges into the top of the lower essential spectrum when $\theta=\pi+\gamma$. As for $\theta\in [\pi+\gamma,2\pi+\gamma]$, there is no discrete spectrum for $\wt{H}^{\rm 1D}(m,\theta;\gamma)$ (any would-be eigenfunction is unnormalisable). 

The above spectral computation is quite elementary, and had already been observed to some extent, e.g. in \cite{Witten,BGLM}. Later, we will see that the above spectral flow across $0$ is \emph{topological}, not spurious.

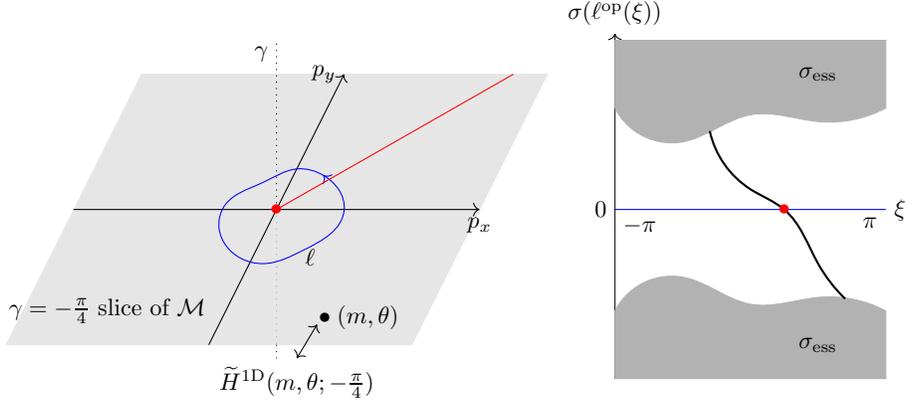
\begin{figure}[h]
\begin{center}

\begin{tikzpicture}[scale=0.9, every node/.style={scale=0.9}]


\filldraw[fill=gray!20, draw=gray!20] (-4,-2) -- (2,-2) -- (4,2) -- (-2,2);
\draw[->] (-3,0)--(3,0);
\draw[->] (-1,-2)--(1,2);
\draw[dotted] (0,0) -- (0,2.5);
\draw[dotted, draw=gray!60] (0,0) -- (0,-2);
\draw[dotted] (0,-2) -- (0,-2.2);
\draw[blue] (1,0) to [closed, curve through = {(0.7,0.5) (0.3,0.6) (-0.2,0.4) (-0.8,0) (-0.2,-0.8) (0.4,-0.6)}] (1,0);
\draw[blue,->] (0.7,0.5) to (0.69,0.51);
\draw[red] (0,0) -- (3.5,2);
\node[red] at (0,0) {$\bullet$};
\node[below] at (3,0) {$p_x$};
\node[left] at (1,2) {$p_y$};
\node at (0.5,-0.7) {$\ell$};
\node[left] at (0,2.3) {$\gamma$};
\node at (-2.5,-1.5) {$\gamma=-\frac{\pi}{4}$ slice of $\mathcal{M}$};
\node[right] at (0.5,-1.6) {$\bullet\;(m,\theta)$};
\draw[<->] (0.6,-1.7) -- (0.3,-2.2);
\node[below] at (0.3,-2.2) {$\widetilde{H}^{\rm 1D}(m,\theta;-\frac{\pi}{4})$};

\draw[blue] (5,0) -- (9,0);
\draw[->] (5,-2.5) -- (5,2.6);
\filldraw[gray!60] (5,-1.5) to [curve through = {(6,-1) (7,-1.4) (8,-1.3)}] (9,-1.5) -- (9,-2.5) -- (5,-2.5) -- (5,-1.5);
\filldraw[gray!60] (5,1.5) to [curve through = {(6,1) (7,1.4) (8,1.3)}] (9,1.5) -- (9,2.5) -- (5,2.5) -- (5,1.5);
\draw[thick] (6.4,1.16) to [curve through = {(7,0.3) (7.5,0) (8,-0.8)}] (8.4,-1.32);
\node[red] at (7.5,0) {$\bullet$}; 
\node[below right] at (5,0) {$-\pi$};
\node[below left] at (9,0) {$\pi$};
\node[left] at (5,0) {$0$};
\node at (8,2) {$\sigma_{\rm ess}$};
\node at (8,-2) {$\sigma_{\rm ess}$};
\node[above] at (5,2.6) {$\sigma(\ell^{\rm op}(\xi))$};
\node[right] at (9,0) {$\xi$};

\end{tikzpicture}
\caption{A constant $\gamma$ slice in the space $\mathcal{M}\cong\widehat{\RR}^2\times{\rm U}(1)$ of half-line Dirac Hamiltonians $\wt{H}^{\rm 1D}(m,\theta;\gamma)$. A loop $\ell$ winding around $(0,0)\in\widehat{\RR}^2$ is shown (blue curve). The spectrum of the corresponding loop $\ell^{\rm op}$ of half-line Dirac Hamiltonians is plotted on the right diagram, exhibiting spectral flow $-1$ across 0. The spectral flow of $-1$ is stable against changes in boundary conditions $\gamma$ and small perturbations. By varying the choice of loop in $\widehat{\RR}^2\setminus \{(0,0)\}$, the Fermi arc locus (red ray) is obtained.}\label{fig:Dirac.spectral.flow}
\end{center}
\end{figure}

\section{Spectral flow structure of half-line Dirac Hamiltonians}\label{sec:spectral.flow.structure}
\subsection{Generalities on spectral flow}\label{sec:spectral.flow.generalities}
{\bf Bounded case.} Along a \emph{norm-continuous} path $B_t, t\in[0,1]$, of \emph{bounded} self-adjoint Fredholm operators, 0 is never in the essential spectrum of $B_t$, but might be an eigenvalue for some of the $B_t$. Intuitively, the eigenvalues of $B_t$ near 0 depend continuously on $t$, so that we can conceive of  the \emph{spectral flow} of $\{B_t\}_{t\in[0,1]}$ as the signed number of eigenvalues (counted with multiplicity) passing through zero in the negative-to-positive sense, as $t$ is increased from $0$ to $1$ (see right diagrams of Fig.\ \ref{fig:Dirac.spectral.flow} and Fig.\ \ref{fig:bent.Fermi.arc} for sketches). Formally, one inspects the graph $y(t)=\sigma(B_t)$ of a suitably perturbed $\{B_t\}_{t\in [0,1]}$, then takes the intersection number with $y=0$, as in \S7 of \cite{APS}. A precise analytic approach involves partitioning $[0,1]$ into subintervals, and summing up the dimension change of the small negative-eigenvalue eigenspaces of $B_t$ at the initial and end points of each subinterval \cite{Phillips}. 

In \cite{AS}, it was found that the subset $\mathcal{F}^{\rm sa}_*$ of self-adjoint Fredholm operators with \emph{both} positive and negative essential spectrum, provided a classifying space for the $K^{-1}$ functor of topological $K$-theory. Thus $\pi_1(\mathcal{F}^{\rm sa}_*)\cong [S^1,{\rm U}(\infty)]=K^{-1}(S^1)\cong\ZZ$, where ${\rm U}(\infty)$ denotes Bott's infinite unitary group. This integer can actually be identified with the spectral flow of a loop in $\mathcal{F}^{\rm sa}_*$ as follows. Following Prop.\ 4 of \cite{Phillips}, one first passes from $\mathcal{F}^{\rm sa}_*$ to the ``spectrally-flattened'' space
\begin{equation*}
\hat{F}^\infty_*:=\{B\in\mathcal{F}^{\rm sa}_*\;|\;||B||=1,\;\;\sigma(B)\;{\rm is\,\,finite,}\;\;\sigma_{\rm ess}(B)=\{-1,1\}\}
\end{equation*}
by picking a small spectral interval $[-\delta,\delta]$ of $B$ containing only finitely many eigenvalues, and then pushing the spectra above $\delta$ (resp.\ below $-\delta$) to $+1$ (resp.\ $-1$). Different $B\in \mathcal{F}^{\rm sa}_*$ may require different $\delta$ to do the spectral-flattening, but a partition-of-unity argument shows that this procedure can be patched up into a well-defined map $\phi:B\mapsto \phi(B)$ which gives a homotopy equivalence $\mathcal{F}^{\rm sa}_*\simeq \hat{F}^\infty_*$. Following this, the map 
\begin{equation*}
\hat{F}^\infty_*\ni B\mapsto {\rm exp}(i\pi(B+1)) \in {\rm U}(\infty)
\end{equation*}
is also a homotopy equivalence \cite{AS, Phillips}. Finally, the following diagram of isomorphisms commutes \cite{Phillips}:
\begin{equation}
\xymatrix{
\pi_1(\mathcal{F}^{\rm sa}_*) \ar[r]^{\rm sf} \ar[d]_{\phi_*} & \ZZ  \\
\pi_1(\hat{F}^{\rm sa}_*) \ar@<-.8ex>[u]_{i_*} \ar[r]^{\rm exp} &  {\rm U}(\infty) \ar[u]_{\rm Wind}
}.\label{eqn:sf.commuting.diagram}
\end{equation}

{\bf Unbounded case.} For loops (or just paths) of self-adjoint \emph{unbounded} Fredholm operators $\mathcal{F}^{\rm sa}$, spectral flow is homotopy-invariant when $\mathcal{F}^{\rm sa}$ is given the \emph{Riesz} topology (norm-continuity of bounded transform) or \emph{gap topology} (norm-continuity of resolvents $(H\pm i)^{-1}$), see \cite{BLP, Kato}. Generally speaking, Riesz continuity is difficult to establish, while gap continuity is somewhat more manageable, and is satisfied by loops of Dirac-type operators on compact manifolds-with-boundary parametrised by boundary conditions. Spectral flow is also well-defined for the weaker \emph{Wahl} topology \cite{Wahl}, which is a slightly stronger version of strong-resolvent convergence. We use the Wahl topology at one technical step later (\S\ref{sec:Wahl.Dirac}), because gap-continuity fails for the family of (twisted) Dirac operators on the \emph{non-compact} half-line $\RR_+$. In all these topologies, the self-adjoint Fredholms $\mathcal{F}^{\rm sa}$ (without restriction on $\sigma_{\rm ess}$) give a classifying space for the $K^{-1}$ functor, \cite{Joachim,Wahl}, just as $\mathcal{F}^{sa}_*$ does in the bounded case. Thus spectral flow of an operator loop in $\mathcal{F}^{\rm sa}$ can be identified with the $K^{-1}(S^1)\cong[S^1,\mathcal{F}^{\rm sa}]\cong\ZZ$ class given by its homotopy class in $\pi_1(\mathcal{F}^{\rm sa})$.

\subsection{Continuous loops of Dirac Hamiltonians}\label{sec:continuity.loop.Dirac}

Any gap-continuous self-adjoint Fredholm operator loop $\ell^{\rm op}:S^1\rightarrow \mathcal{M}^\prime\subset\mathcal{M}$, has a homotopy invariant spectral flow, denoted ${\rm sf}(\ell^{\rm op})$.
 The remainder of this section explores the spectral flow structure of (loops inside) $\mathcal{M}^\prime$.

\subsubsection{Warm-up: Loops of half-line Dirac Hamiltonians with fixed boundary condition}\label{sec:basic.loop}
\begin{lemma}\label{lem:gap.continuous.Dirac}
For any fixed boundary condition $\gamma\in {\rm U}(1)$, the self-adjoint family
\begin{equation*}
\widehat{\RR}^2\ni (m,\theta)\mapsto \wt{H}^{\rm 1D}(m,\theta;\gamma)
\end{equation*}
is gap-continuous.
\end{lemma}
\begin{proof}
$\wt{H}^{\rm 1D}(m,\theta;\gamma)$ is obtained from the $m=0$ case by simply adding a bounded mass term which depends norm-continuously on $(m,\theta)\in\widehat{\RR}^2$. 
\end{proof}

\begin{proposition}\label{prop:basic.loop.sf}
Let $\ell:S^1\rightarrow \widehat{\RR}^2\setminus \{(0,0)\}, \xi\mapsto (m(\xi),\theta(\xi))$ be a continuous loop avoiding the origin. For each fixed $\gamma\in {\rm U}(1)$, the corresponding operator loop 
\begin{equation*}
\ell^{\rm op}:S^1\ni e^{i\xi}\mapsto \wt{H}^{\rm 1D}(m(\xi),\theta(\xi);\gamma)
\end{equation*}
is a gap-continuous loop in $\mathcal{F}^{\rm sa}$, whose spectral flow is the winding number,
\begin{equation*}
{\rm sf}(\ell^{\rm op})=-{\rm Wind}(e^{i\xi}\mapsto e^{i\theta(\xi)}).
\end{equation*}
\end{proposition}
\begin{proof}
Gap-continuity of $\ell^{\rm op}$ is given by Lemma \ref{lem:gap.continuous.Dirac}. In particular, at any $m>0$, the basic loop 
\begin{equation*}
\ell^{\rm op}_{\rm basic}:S^1\ni e^{i\xi}\mapsto \wt{H}^{\rm 1D}(m,\xi;\gamma),
\end{equation*}
has spectral flow of $-1$, from the calculations, Eq. \eqref{eqn:1D.evalue}-\eqref{eqn:1D.efunction}, of \S\ref{sec:massive.halfline.Dirac.basic}. For a general $\ell$ with winding number $n$, it is homotopic to $e^{i\xi}\mapsto me^{in\xi}$, so that the homotopy class of $\ell^{\rm op}$ in $\pi_1(\mathcal{F}^{\rm sa})$ is similarly $n$ times that of $\ell^{\rm op}_{\rm basic}$. Thus 
\begin{equation*}
{\rm sf}(\ell^{\rm op})=n\cdot{\rm sf}(\ell^{\rm op}_{\rm basic})=-n=-{\rm Wind}(\ell)=-{\rm Wind}(e^{i\xi}\mapsto e^{i\theta(\xi)}).
\end{equation*}
\end{proof}
Note that the ${\rm Wind}(\ell)$ does not depend on the radial coordinate $m$.

\subsubsection{Loops with varying boundary conditions}\label{sec:momentum.dependent.loop}
Observe that $\wt{H}^{\rm 1D}(m,\theta;\gamma)$ is unitarily equivalent to $\wt{H}^{\rm 1D}(m,\theta-\gamma;0)$ via conjugation by the constant unitary matrix function $U_\gamma:z\mapsto {\rm diag}(1,e^{i\gamma})$: the mass term transforms as
\begin{equation*}
\begin{pmatrix}1 & 0 \\ 0 & e^{i\gamma} \end{pmatrix}\begin{pmatrix}0 & me^{-i(\theta-\gamma)} \\ me^{i(\theta-\gamma)} & 0 \end{pmatrix}\begin{pmatrix}1 & 0 \\ 0 & e^{-i\gamma} \end{pmatrix}=\begin{pmatrix}0 & me^{-i\theta} \\ me^{i\theta} & 0 \end{pmatrix},
\end{equation*}
while the boundary condition parameter is shifted from $0$ to $\gamma$ after applying $U_\gamma$.

\begin{corollary}\label{cor:general.Dirac.loop.sf}
For a general continuous loop 
\begin{equation*}
\ell:S^1\rightarrow\widehat{\RR}^2\times{\rm U}(1),\quad\xi\mapsto (m(\xi),\theta(\xi);\gamma(\xi)),
\end{equation*}
the corresponding loop of half-line Dirac Hamiltonians
\begin{equation*}
\ell^{\rm op}:S^1\rightarrow \mathcal{M}^\prime, \qquad \xi\mapsto \wt{H}^{\rm 1D}(m(\xi),\theta(\xi);\gamma(\xi)),
\end{equation*}
is gap-continuous. If $(m,\theta)$ avoids the origin, then $\ell^{\rm op}$ has spectral flow 
\begin{equation*}
{\rm sf}(\ell^{\rm op})=-{\rm Wind}(e^{i\xi}\mapsto e^{i(\theta(\xi)-\gamma(\xi))}).
\end{equation*}
\end{corollary}
\begin{proof}
Note that $\gamma$ is continuous, and that $\ell^{\rm op}$ is obtained from the loop
\begin{equation*}
\hat{\ell}^{\rm op}:\xi\mapsto \wt{H}^{\rm 1D}(m(\xi),\theta(\xi)-\gamma(\xi);0)
\end{equation*} 
by pointwise conjugation with the norm-continuous unitary loop $\xi\mapsto U_{\gamma(\xi)}$.
Thus $\ell^{\rm op}$ is another gap-continuous loop, and is a self-adjoint Fredholm loop if $\ell$ avoids the origin of $\widehat{\RR}^2$. In that case, Prop.\ \ref{prop:basic.loop.sf} gives
\begin{equation*}
{\rm sf}(\ell^{\rm op})={\rm sf}(\hat{\ell}^{\rm op})=-{\rm Wind}(e^{i\xi}\mapsto e^{i(\theta(\xi)-\gamma(\xi))}).
\end{equation*}
\end{proof}

\section{The topological Fermi arcs of Weyl Hamiltonians}\label{sec:Weyl.Fermi}

\subsection{From half-space Weyl Hamiltonian to half-line Dirac Hamiltonians: boundary conditions}

The Weyl Hamiltonian, Eq.\ \eqref{eqn:Weyl.Hamiltonian}, is self-adjoint on the domain $H^1(\RR^3)^{\oplus 2}\subset L^2(\RR^3)^{\oplus 2}$. Now consider the Weyl Hamiltonian $\wt{H}^{\rm Weyl}$ on the upper half Euclidean space (or simply \emph{half-space}), $\RR^2\times\RR_+$, where $\RR_+:=(0,\infty)$. This is at first defined as a symmetric operator on $C_0^\infty(\RR^2\times\RR_+)^{\oplus 2}$, but is not essentially self-adjoint. Indeed one sees from an integration-by-parts that for any pair of differentiable wavefunctions $\varphi=(\varphi_1,\varphi_2), \psi=(\psi_1, \psi_2)$ in a domain of true self-adjointness, it is necessary that the boundary integral vanishes,
\begin{equation}
\int_{\RR^2\times\{0\}} (\overline{\varphi_1}\psi_1-\overline{\varphi_2}\psi_2)(x,y,0)\,dx\,dy=\int_{\RR^2\times\{0\}} (\varphi^\dagger\sigma_z\psi)(x,y,0)\,dx\,dy=0. \label{eqn:boundary.integral}
\end{equation}

We will only consider \emph{homogeneous} boundary conditions, i.e., independent of $(x,y,0)$ on the boundary plane. We can achieve Eq.\ \eqref{eqn:boundary.integral} by requiring that
\begin{equation}
\psi^\dagger\sigma_z\psi =0 \;\;\;{\rm at}\;\;z=0,\label{eqn:angular.momentum.reversal}
\end{equation}
(and also for $\varphi$). This is the statement that the spinors $\psi(x,y,0)$ are spin-polarised parallel to the boundary\footnote{A physical way to understand why boundary conditions must be constrained as above, is given by Witten, \S 1.10 of \cite{Witten}. Namely, the \emph{helicity}, or \emph{chirality}, of Weyl spinors is constrained such that the angular momentum is parallel to the direction of motion. So if the linear momentum $p_z$ is reversed after impinging on the boundary (to conserve probability), then the $z$-angular momentum must likewise be reversed, which is Eq.\ \eqref{eqn:angular.momentum.reversal}.}. 

With homogeneous boundary conditions, there remains $\RR^2$ translation invariance in the $x$-$y$ variables, so we can Fourier transform them into $(p_x,p_y)\in\widehat{\RR}^2$ where $\widehat{\RR}^2$ is the Pontryagin dual of $\RR^2$. Physically, $\widehat{\RR}^2$ is the conserved momenta parallel to the boundary $z=0$.
Then $\wt{H}^{\rm Weyl}$ formally Fourier transforms into the $\widehat{\RR}^2$-parametrised family of half-line Dirac Hamiltonians
\begin{equation*}
\wt{H}^{\rm Weyl}\overset{\rm Fourier}{\longrightarrow} \{\wt{H}^{\rm 1D}(m,\theta))\}_{(m,\theta)\in\widehat{\RR}^2},
\end{equation*}
where we made a switch to polar coordinates for $\widehat{\RR}^2$. Upon introducing a continuous \emph{momentum-dependent family} $\gamma:(m,\theta)\mapsto \gamma(m,\theta)\in{\rm U}(1)$ of self-adjoint boundary conditions, we have precisely a decomposition
\begin{equation}
\wt{H}^{\rm Weyl}(\gamma)\overset{\rm Fourier}{\longrightarrow} \{\wt{H}^{\rm 1D}(m,\theta;\gamma(m,\theta))\}_{(m,\theta)\in\widehat{\RR}^2}\label{eqn:Weyl.decomposed.general}
\end{equation}
into half-line Dirac Hamiltonians, encountered earlier in Eq.\ \eqref{eqn:half.line.Dirac.domain}.

\begin{definition}\label{defn:Fermi.arc}
Let $F:X\rightarrow \mathcal{F}^{sa}$ be a family of unbounded self-adjoint Fredholm operators parametrised by a set $X$. The \emph{Fermi arc} of the family is the subset 
\begin{equation*}
l_{\rm Fermi}=\{x\in X\;|\; 0\in\sigma(F(x))\}.
\end{equation*}
\end{definition}
In physics examples, $X$ will be a two-dimensional manifold, and $F$ will be a continuous family in some topology on $\mathcal{F}^{\rm sa}$. On each path $[0,1]\ni \xi\mapsto F(\xi)$, a generically finite set of isolated $\xi$ will have the $F(\xi)$ possessing a zero eigenvalue, then $l_{\rm Fermi}$ is 1-dimensional, hence the terminology ``arc''.

\begin{example}[Basic Fermi arc]\label{ex:basic.Fermi}
In case $\gamma$ is just a \emph{constant} function, $\wt{H}^{\rm Weyl}$ gives rise to the family $\{\wt{H}^{\rm 1D}(m,\theta;\gamma)\}_{(m,\theta)\in \widehat{\RR}^2}$, and on $X=\widehat{\RR}^2\setminus \{(0,0)\}$, this is a self-adjoint Fredholm family. 
By the calculations in \S\ref{sec:massive.halfline.Dirac.basic}, a zero eigenvalue occurs only for the operators $\wt{H}^{\rm 1D}(m,\gamma+\frac{\pi}{2};\gamma)$. Thus, the Fermi arc is the ray $\theta=\gamma+\frac{\pi}{2}$, see Fig.\ \ref{fig:Dirac.spectral.flow}. Notice that the momentum on the arc points in a direction which is $\frac{\pi}{2}$ anticlockwise-shifted from the surface spin polarisation angle $\gamma$ --- this is \emph{spin-momentum locking}.
\end{example}

\begin{remark}
One might be tempted to assume a momentum-independent boundary condition $\gamma$ for $\wt{H}^{\rm Weyl}$, e.g.\ \cite{Witten,HKW,BGLM}. However, in actual experiments involving Weyl semimetals (e.g.\ \cite{Xu, Xu2, Lv, Morali, Souma}), for which $H^{\rm Weyl}$ is an approximation near band crossings, one already observes that the Fermi arcs are bent and have some spin polarisation (or ``spin texture'') which is not constant in momentum space. Therefore, we \emph{must} allow momentum-dependent $\gamma$-functions, and will seek to justify why this does not affect the qualitative features of the Fermi arc.

In full generality, $\gamma$ need only be measurable, see, e.g., \cite{GL} in the context of Fourier transforming half-plane Dirac Hamiltonians, and \cite{GJT}. Typically, in a very symmetric (say rotationally) problem, continuity of $\gamma$ might be lost at special momentum space points such as the origin, and could affect the spectral flow, as we saw from Cor.\ \ref{cor:general.Dirac.loop.sf}. However, in our application to Weyl semimetals, \S\ref{sec:tight.binding.Weyl}, the projected Weyl points are in general position in $\widehat{\RR}^2$, and there is no reason why the spin polarisation boundary condition at those generic points should become discontinuous. This is why we assume that $\gamma$ is continuously-defined everywhere in $\widehat{\RR}^2$.
\end{remark}

\begin{remark}
Abstractly speaking, the symmetric operator $\wt{H}^{\rm Weyl}$ initially defined on $C_0^\infty(\RR^2\times\RR_+)$ has infinite deficiency indices, and so a plethora of self-adjoint extensions. Amongst these, a subclass corresponds to local homogeneous boundary conditions, as specified by some $\gamma:\widehat{\RR}^2\rightarrow{\rm U}(1)$. Within these, there is a special circle of momentum-independent boundary conditions. 
\end{remark}

\subsection{Stability of Fermi arcs against boundary condition modification}
The half-space Weyl Hamiltonian subject to the boundary condition specified by a \emph{continuous} $\gamma:\widehat{\RR}^2\rightarrow{\rm U}(1)$, will be denoted $\wt{H}^{\rm Weyl}(\gamma)$. Its Fourier transform along $x$-$y$ was given in Eq. \eqref{eqn:Weyl.decomposed.general}, rewritten here for convenience,
\begin{equation*}
\wt{H}^{\rm Weyl}(\gamma)=\int_{\widehat{\RR}^2}^\oplus \wt{H}^{1D}(m,\theta;\gamma(m,\theta)).
\end{equation*}
Equivalently, $\wt{H}^{\rm Weyl}(\gamma)$ is specified by the continuous ``classifying map'' 
\begin{equation}
\wt{h}^{\rm Weyl;\gamma}:\widehat{\RR}^2\rightarrow\widehat{\RR}^2\times{\rm U}(1)\cong\mathcal{M},\qquad (m,\theta)\mapsto(m,\theta;\gamma(m,\theta))\leftrightarrow \wt{H}^{1D}(m,\theta;\gamma(m,\theta)).\label{eqn:Weyl.function}
\end{equation}

As before, take any continuous momentum space loop avoiding the origin, $\ell:S^1\rightarrow \widehat{\RR}^2\setminus\{(0,0)\}$. This picks out a gap-continuous Fredholm loop (see Fig.\ \ref{fig:bent.Fermi.arc} for a sketch)
\begin{equation*}
\ell^{\rm op}_{{\rm Weyl};\gamma}:=\wt{h}^{\rm Weyl;\gamma}\circ\ell:\xi\mapsto \left(m(\xi),\theta(\xi);\gamma(\ell(\xi))\right)\leftrightarrow \wt{H}^{1D}\left(m(\xi),\theta(\xi);\gamma(\ell(\xi))\right).
\end{equation*}

\begin{proposition}\label{prop:boundary.independence}
With $\wt{H}^{\rm Weyl}(\gamma)$ a half-space Weyl Hamiltonian as above, let $\ell:S^1\rightarrow \widehat{\RR}^2\setminus\{(0,0)\}$ be a continuous momentum space loop, and 
\begin{equation*}
\ell^{\rm op}_{{\rm Weyl};\gamma}=\wt{h}^{\rm Weyl;\gamma}\circ\ell:\xi\mapsto \wt{H}^{\rm 1D}\left(m(\xi),\theta(\xi);\gamma(\ell(\xi))\right).
\end{equation*}
be the corresponding (gap-continuous, Fredholm) loop of massive half-line Dirac Hamiltonians. Its spectral flow equals the winding number of $\ell$,
\begin{equation}
{\rm sf}(\ell^{\rm op}_{{\rm Weyl};\gamma})=-{\rm Wind}\left(e^{i\xi}\mapsto e^{i(\theta(\xi))}\right),\label{eqn:sf.simple}
\end{equation}
independent of the boundary condition function $\gamma$.
\end{proposition}
\begin{proof}
According to Corollary \ref{cor:general.Dirac.loop.sf}, we should have
\begin{equation}
{\rm sf}(\ell^{\rm op}_{{\rm Weyl};\gamma})=-{\rm Wind}\left(e^{i\xi}\mapsto e^{i(\theta(\xi)-\gamma(\ell(\xi)))}\right).\label{eqn:sf.complicated}
\end{equation}
But $\gamma$ is continuously defined over the entire momentum space $\widehat{\RR}^2$, not just on the image of the loop $\ell$, so that $e^{i\xi}\mapsto e^{-i\gamma(\ell(\xi)))}$ has zero winding. Then Eq.\ \eqref{eqn:sf.complicated} simplifies to Eq. \eqref{eqn:sf.simple}.
\end{proof}

{\bf Bent Fermi arc.} In particular, for each of the concentric momentum space loops at some fixed $m>0$, there is one (nett) contribution to the Fermi arc of $\wt{H}^{\rm Weyl}(\gamma)$ due to the $\gamma$-independent spectral flow of $-1$ across 0 energy. This contribution occurs at those momentum directions $\theta_m$ that solve $\theta_m=\gamma(m,\theta_m)+\frac{\pi}{2}$ (spin-momentum locking). Thus the effect of a general momentum-dependent boundary spin polarisation condition $\gamma$, is to bend the straight Fermi arc that connects the origin of momentum space to infinity, see Fig.\ \ref{fig:bent.Fermi.arc}.

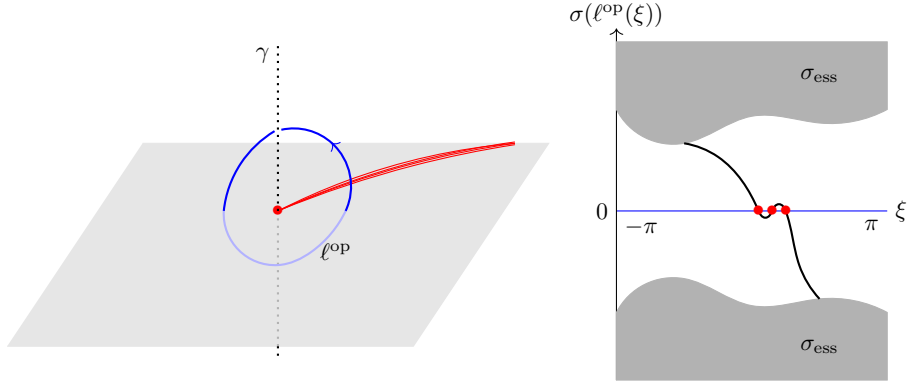
\begin{figure}[h]
\begin{center}

\begin{tikzpicture}[scale=0.9, every node/.style={scale=0.9}]


\filldraw[fill=gray!20, draw=gray!20] (-4,-2) -- (2,-2) -- (4,1) -- (-2,1);
\draw[dotted, thick, draw=gray!60] (0,0) -- (0,-2);
\draw[dotted, thick] (0,-2) -- (0,-2.2);
\draw[blue!30,thick] (-0.8,0) to [curve through = {(0,-0.8) (0.5,-0.6)}] (1,0);
\draw[red] (0,0) to [curve through = {(2,0.7)}] (3.5,1);
\draw[red] (0,0) to [curve through = {(2,0.72)}] (3.5,1.02);
\draw[red] (0,0) to [curve through = {(2,0.67)}] (3.5,0.98);
\draw[red] (0,0) to [curve through = {(1,0.44)}] (3.5,1);
\node[red] at (0,0) {$\bullet$};
\draw[blue,->] (0.8,1) to (0.79,1.01);
\draw[blue,thick] (1,0) to [curve through = {(0.8,1)}] (0.05,1.2);
\draw[blue,thick] (-0.05,1.18) to [curve through = 
 {(-0.55,0.7)}] (-0.8,0);
\draw[dotted,thick] (0,0) -- (0,2.5);

\node[right] at (0.5,-0.6) {$\ell^{\rm op}$};
\node[left] at (0,2.3) {$\gamma$};

\draw[blue] (5,0) -- (9,0);
\draw[->] (5,-2.5) -- (5,2.7);
\filldraw[gray!60] (5,-1.5) to [curve through = {(6,-1) (7,-1.4) (8,-1.3)}] (9,-1.5) -- (9,-2.5) -- (5,-2.5) -- (5,-1.5);
\filldraw[gray!60] (5,1.5) to [curve through = {(6,1) (7,1.4) (8,1.3)}] (9,1.5) -- (9,2.5) -- (5,2.5) -- (5,1.5);
\draw[thick] (6,1) to [curve through = {(6.5,0.8) (7.1,0) (7.2,-0.1) (7.3,0) (7.4,0.1) (7.5,0) (7.7,-0.8)}] (8,-1.3);
\node[red] at (7.1,0) {$\bullet$}; 
\node[red] at (7.3,0) {$\bullet$}; 
\node[red] at (7.5,0) {$\bullet$}; 
\node[below right] at (5,0) {$-\pi$};
\node[below left] at (9,0) {$\pi$};
\node[left] at (5,0) {$0$};
\node at (8,2) {$\sigma_{\rm ess}$};
\node at (8,-2) {$\sigma_{\rm ess}$};
\node[above] at (5,2.6) {$\sigma(\ell^{\rm op}(\xi))$};
\node[right] at (9,0) {$\xi$};

\end{tikzpicture}

\caption{The blue curve represents an operator loop $\ell^{\rm op}:\xi\mapsto \wt{H}^{\rm 1D}(\ell(\xi);\gamma(\xi))$, similar to that in Fig.\ \ref{fig:Dirac.spectral.flow} except that the boundary condition $\gamma$ is not assumed to be constant throughout the loop. Small perturbations of $\ell^{\rm op}(\xi)$ are also allowed. The spectral flow remains $-1$, as sketched in the right diagram, but is not generally monotone. So there may be several contributions (the three red dots) to the Fermi arc from such a loop $\ell$. This results in a fuzzy and bent Fermi arc (red fuzzy curve in left diagram).) }\label{fig:bent.Fermi.arc}
\end{center}
\end{figure}

\subsection{Stability of Fermi arcs against chemical potential changes}\label{sec:chemical.potential}
We can also consider spectral flow across $\mu\neq 0$, which is the same thing as shifting $H^{\rm Weyl}$ to  $H^{\rm Weyl}-\mu$, and taking spectral flow of $\wt{H}^{\rm 1D}(m,\theta;\gamma)-\mu$ across $0$. In the physics context, this is a modification of the \emph{chemical potential/Fermi level}. The continuity of $\wt{H}^{\rm 1D}(m,\theta;\gamma)-\mu$ with respect to the variables $m,\theta;\gamma$ is unaffected, but now Fredholmness is only achieved when $m>|\mu|$. Spectral flow (across 0) now only makes sense for large loops staying at $m> |\mu|$, and remains equal to that for the unshifted loop. The Fermi arc now emerges out of some $(m=\mu,\theta)\in\widehat{\RR}^2$ rather than exactly from $(m=0,0)$.

\subsection{Stability of Fermi arcs against potential terms}\label{sec:stability.potential}
We may add to $\wt{H}^{\rm Weyl}(\gamma)$ some extra homogeneous smooth potential term $V$, where $V=V(x,y,z)=V(z)$ is $2\times 2$ Hermitian matrix-valued and independent of $x,y$. Call the perturbed half-space Weyl Hamiltonian $\wt{H}^{\rm Weyl}_V(\gamma)=\wt{H}^{\rm Weyl}(\gamma)+V$ (this remains self-adjoint on the same domain as the unperturbed one \cite{Kato}). Such a perturbation is natural, since apart from a boundary condition $\gamma$, there should typically also be some confining potential $V$ which decays away from the boundary. After Fourier transforming in $x,y$, we get the self-adjoint family $(m,\theta)\mapsto\wt{H}^{1D}_V(m,\theta;\gamma)=\wt{H}^{\rm 1D}(m,\theta;\gamma)+V$, where $V=V(z)$ may be regarded as a bounded perturbation operator on $L^2(\RR_+)^{\oplus 2}$.

\begin{theorem}\label{thm:perturbation.independence}
Let $H^{\rm Weyl}_V(\gamma)=\wt{H}^{\rm Weyl}(\gamma)+V$ be the perturbed half-space Weyl Hamiltonian with $V=V(z)$ a $2\times 2$ Hermitian matrix-valued smooth potential vanishing at infinity. The spectral flow structure for $H^{\rm Weyl}_V(\gamma)$ is independent of $V$.
\end{theorem}
\begin{proof}
A fairly standard perturbation theory argument shows that $V$ is $\wt{H}^{\rm 1D}(m,\theta;\gamma)$-relatively compact, see Appendix \ref{sec:perturbations} for details. This means that the essential spectrum of $\wt{H}^{\rm 1D}_V(m,\theta;\gamma)$ remains unchanged by adding $V$ (Theorem 5.35, \S IV of \cite{Kato}, Corollary 2, VIII.4 of \cite{RS4}). As in Prop.\ \ref{prop:boundary.independence}, let $\ell:S^1\rightarrow \widehat{\RR}^2\setminus\{(0,0)\}$ be a continuous momentum space loop; the corresponding loop of perturbed massive Dirac Hamiltonians
\begin{equation*}
\ell^{\rm op}_{{\rm Weyl}, V;\gamma}:\xi\mapsto \wt{H}^{\rm 1D}_V\left(m(\xi, \theta(\xi); \gamma(\ell(\xi))\right)
\end{equation*}
remains Fredholm and gap-continuous. By ``turning off $V$'', we can homotope $\ell^{\rm op}_{{\rm Weyl}, V;\gamma}$ to the unperturbed $\ell^{\rm op}_{{\rm Weyl};\gamma}$ of Prop.\ \ref{prop:boundary.independence}, preserving the spectral flow.
\end{proof}

\begin{remark}
Without attempting to be technically precise, we mention that we could allow worse perturbations which enlarge the essential spectrum slightly. Then only the operator loops at large enough $m$ (which had a large essential spectrum gap when the perturbation was absent) remain Fredholm loops. The Fermi arc only emerges out of an essential gap closing point at some $m>0$.
\end{remark}

\vspace{1em}
\begin{remark}\label{rem:fuzzy.arc}
{\bf Fuzziness of perturbed Fermi arc.}
With the potential term $V$, it becomes harder (generally impossible) to exactly solve for the spectrum of the various $\wt{H}^{1D}_V(m,\theta;\gamma)$. \emph{This is where the topological invariance of spectral flow becomes indispensable} --- we can still be sure that every loop (at large enough $m$) winding around the origin once will have spectral flow $-1$. This means that \emph{counted with signs}, one value of $\theta$ will contribute to the Fermi arc for each large $m>0$. If we do not count signs, there can generally be several values of $\theta$ contributing to the Fermi arc (the spectral flow need not be monotone in $\theta$, unlike the exactly solvable $V\equiv 0$ case), see Fig.\ \ref{fig:bent.Fermi.arc}. Combining the effect of $V$ with a shift of chemical potential (as in \S\ref{sec:chemical.potential}), we see that the Fermi arc is predicted to be a rather fuzzy arc emerging from an imprecisely-defined point close to the origin, as experiments show!
\end{remark}

\subsection{Gauge invariance and momentum shift of Fermi arcs}\label{sec:gauge.invariance}
To model general \emph{pairs} of Weyl points later on (\S\ref{sec:shifted.Weyl.points}), we will need to consider the Dirac operators $H^{\rm 1D}(m,\theta)$ coupled to a (necessarily flat) ${\rm U}(1)$-connection $\mathcal{A}$. Geometrically, this means that we twist (tensor) the spinor bundle $\RR\times\CC^2$ with the line bundle $\mathcal{L}_\mathcal{A}$ with connection $\mathcal{A}$. We assume, for simplicity, that the connection is translation invariant, thus of the form $\mathcal{A}=A\,dz$ for some constant $A\in\RR$ (the general case is discussed in Remark \ref{rem:general.gauge}). Then the $\mathcal{L}_\mathcal{A}$-twisted, or $A$-covariant, Dirac operator on the line, is a bounded perturbation of the untwisted one,
\begin{equation*}
H^{\rm 1D}_A(m,\theta):=\begin{pmatrix} -i\frac{d}{dz}-A & me^{-i\theta} \\ me^{i\theta} & i\frac{d}{dz}+A\end{pmatrix}. 
\end{equation*}
Were we to Fourier transform the $z$-variable, the $p_z$ momentum will simply be shifted to $p_z-A$. The unitary gauge transformation $V_A$ of multiplication by the phase function $z\mapsto e^{iAz}$, intertwines $H^{\rm 1D}_A(m,\theta)$ with the untwisted $H^{\rm 1D}(m,\theta)$, so that they have the same spectrum (only the energy-momentum dispersion is shifted, as above).

There is a circle family of self-adjoint half-line versions of $H^{\rm 1D}_A(m,\theta)$, labelled by the spin polarisation angle $\gamma$ at the boundary $z=0$. We denote these by $\wt{H}^{\rm 1D}_A(m,\theta;\gamma)$. The gauge transformation $V_A$ effects the same ${\rm U}(1)$ phase shift on both components of the spinor, so it intertwines $\wt{H}^{\rm 1D}_A(m,\theta;\gamma)$ with $\wt{H}^{\rm 1D}(m,\theta;\gamma)$, without changing the boundary condition label.

Including the extra twisting variable $A\in\RR$, the 3-parameter space $\mathcal{M}=\widehat{\RR}^2\times{\rm U}(1)$ of untwisted half-line Dirac Hamiltonians is expanded to a 4-parameter space of \emph{twisted} half-line Dirac Hamiltonians $\mathscr{M}\cong\widehat{\RR}^2\times{\rm U}(1)\times\RR$ inside the unbounded self adjoint operators. Within $\mathscr{M}$, there is a 4-parameter subset $(\widehat{\RR}^2\setminus\{(0,0)\})\times{\rm U}(1)\times\RR\cong\mathscr{M}^\prime\subset\mathcal{F}^{\rm sa}$ of massive twisted half-line Dirac Hamiltonians.

Unlike $\mathcal{M}$, however, the family $\mathscr{M}$ (thus also $\mathscr{M}^\prime$) does \emph{not} depend gap-continuously on the twisting variable $A$, see \S\ref{sec:Wahl.Dirac}, thus spectral flow of loops in $\mathscr{M}^\prime$ is ill-defined in the sense of \cite{BLP}. In the next section \S\ref{sec:Wahl.Dirac}, we explain how to fix this by working in a weaker operator topology.

\subsubsection{Wahl topology continuity of twisted half-line Dirac Hamiltonians}\label{sec:Wahl.Dirac}
Notice that $\{V_A\}_{A\in\RR}$ is a strongly-continuous 1-parameter family of unitaries, generated by the position operator $\hat{z}$. Thus, variation of $\wt{H}^{\rm 1D}_A(m,\theta;\gamma)\in\mathscr{M}$ in the $A$-variable is effected in a strongly-continuous way. Thus
\begin{equation*}
\mathscr{M}=\{\wt{H}^{\rm 1D}_A(m,\theta;\gamma)\}_{(m,\theta)\in\widehat{\RR}^2, \gamma\in{\rm U}(1), A\in\RR}
\end{equation*}
is at least a strong-resolvent continuous family of self-adjoint operators, jointly in the four variables $m,\theta,\gamma,A$. 

Recall that the conditions for a general closed operator path $\ell^{\rm op}$ to be Wahl-continuous, are that
\begin{itemize}
\item[(i)] the path of resolvents is strongly-continuous,
\item[(ii)] there exists $\epsilon>0$ and a smooth real-valued even ``bump function'' $\phi$ supported in $[-\epsilon,\epsilon]\subset\RR$ with $\phi^\prime|_{(-\epsilon,0)}>0$, such that $\xi\mapsto\phi(\ell^{\rm op}(\xi))$ is a norm-continuous path. 
\end{itemize}

Let $\ell^{\rm op}$ be a strong-resolvent continuous operator path, 
\begin{equation}
\ell^{\rm op}:[0,1]\ni\xi\mapsto \wt{H}^{\rm 1D}_{A(\xi)}(m(\xi),\theta(\xi);\gamma(\xi))\in\mathscr{M}^\prime,\label{eqn:twisted.Dirac.loop}
\end{equation}
and
\begin{equation*}
\hat{\ell}^{\rm op}:[0,1]\ni\xi\mapsto \wt{H}^{\rm 1D}(m(\xi),\theta(\xi);\gamma(\xi))\in\mathcal{M}^\prime\subset\mathscr{M}^\prime
\end{equation*}
be its untwisted version (with $A$ set to 0), which is actually gap-continuous due to Cor.\ \ref{cor:general.Dirac.loop.sf}, thus also Wahl-continuous (see \cite{Wahl}, this follows from Theorem VIII.20 of \cite{RS1}). Since $\ell^{\rm op}$ is obtained from $\hat{\ell}^{\rm op}$ by conjugation with the strongly-continuous path of unitary gauge transformations $\{V_{A(\xi)}\}_{\xi\in[0,1]}$, the former is also Wahl continuous \cite{Wahl}, with the same (well-defined) spectral flow as $\hat{\ell}^{\rm op}$. The Wahl continuity of $\ell^{\rm op}$ can also be seen directly by observing that each $\ell^{\rm op}(\xi)$ is a twisted half-line Dirac Hamiltonian with at most one eigenvalue in the essential spectrum gap, with eigenfunction $V_{A(\xi)}\cdot\psi_{m(\xi),\theta(\xi);\gamma(\xi)}$ (given by Eq. \eqref{eqn:1D.efunction}), and that the corresponding rank-1 eigenprojections vary continuously with $A$.

Thus, in the same way as Cor.\ \ref{cor:general.Dirac.loop.sf} and Prop.\ \ref{prop:boundary.independence}, we obtain
\begin{proposition}\label{prop:Wahl.twisted.sf}
Let $\ell:S^1\rightarrow\widehat{\RR}^2\setminus\{(0,0)\}$ be a continuous momentum space loop, $\gamma:\widehat{\RR}^2\rightarrow{\rm U}(1)$ be a continuous boundary condition function, and $A:\widehat{\RR}^2\rightarrow\RR$ be a continuous twisting function. The corresponding loop of twisted massive half-line Dirac Hamiltonians,
\begin{equation*}
\ell^{\rm op}_{\gamma,A}:\xi\mapsto \wt{H}^{\rm 1D}_{A(\ell(\xi))}(\ell(\xi);\gamma(\ell(\xi)))\equiv \wt{H}^{\rm 1D}_{A(\ell(\xi))}(m(\xi),\theta(\xi);\gamma(\ell(\xi))),
\end{equation*}
is a Wahl-continuous self-adjoint Fredholm loop, whose spectral flow is
\begin{equation*}
{\rm sf}(\ell^{\rm op}_{\gamma,A})=-{\rm Wind}(e^{i\xi}\mapsto e^{i\theta(\xi)}),
\end{equation*}
independently of $\gamma$ and $A$.
\end{proposition}

\begin{remark}
As in Theorem \ref{thm:perturbation.independence}, a perturbation term $V$ may also be added, without modifying the spectral flow structure in Prop.\ \ref{prop:Wahl.twisted.sf}.
\end{remark}
\begin{remark}\label{rem:general.gauge}
We could also allow more general $z$-dependent gauge fields $\mathcal{A}=\mathcal{A}(z)\,dz$. The spectral flow along a path of $\wt{H}^{\rm 1D}_\mathcal{A}(m,\theta;\gamma)$ remains well-defined, provided the gauge transformations $V_\mathcal{A}$ vary strongly-continuously in the path parameter. 
\end{remark}

\begin{example}
Consider a 3D Weyl Hamiltonian coupled to a gauge field 
$\mathcal{A}=\sum_{i=1}^3 \mathcal{A}_i\,dx^i$ with vanishing curvature. Assuming the gauge field is constant in $x,y$, then $\mathcal{A}=A_x\,dx +A_y\,dy +\mathcal{A}_z(z) dz$, with some constants $A_x, A_y$. We can still Fourier transform the half-space version $\wt{H}^{\rm Weyl}_\mathcal{A}(\gamma)$ in the $x,y$-variables. Noting that the effect of $A_x\,dx +A_y\,dy$ is just to shift the momenta $(p_x,p_y)$ to $(p_x-A_x, p_y-A_y)$, we obtain a momentum-shifted family of twisted half-line Dirac Hamiltonians,
\begin{equation*}
\wt{H}^{\rm Weyl}_\mathcal{A}(\gamma)\rightsquigarrow\left(\widehat{\RR}^2\ni (p_x,p_y)\mapsto \wt{H}^{\rm 1D}_{\mathcal{A}_z}(p_x-A_x, p_y-A_y;\gamma(p_x,p_y))\right).
\end{equation*}
In this family, the twist $\mathcal{A}_z$ does not change with $(p_x, p_y)$, so the entire family can be gauge transformed by $V_{-\mathcal{A}_z}$ to remove the twist. Then all the arguments about Fermi arcs follow in the same way from Prop.\ \ref{prop:Wahl.twisted.sf}, except that the origin in momentum space is shifted to $(A_x, A_y)$. The spin-momentum locking condition for the Fermi arc is $\theta_\mathcal{A}=\gamma+\frac{\pi}{2}$, with $\theta_\mathcal{A}$ now denoting the angular variable taken with respect to the shifted origin $(A_x, A_y)$.
\end{example}

\subsection{Examples of spurious Fermi arcs and gapless edge states}\label{sec:spurious.Fermi}
\begin{example}\label{ex:spurious.Fermi}
The massless Dirac Hamiltonian in 3D is the direct sum of the left-handed and right-handed Weyl Hamiltonians. Clearly the half-space 3D Dirac Hamiltonian has two Fermi arcs, one coming from each Weyl Hamiltonian, although they are associated with equal and opposite spectral flows. Thus it is possible to have zero nett spectral flow but still have Fermi arcs. The point now is that we can perturb the massless 3D Dirac Hamiltonian, e.g.\ by turning on an off-diagonal mass term, in order to remove the zero eigenvalues. This means that the two Fermi arcs are \emph{spurious}, not \emph{topological}. 

\end{example}

\begin{example}\label{ex:half.Chern} This is essentially the massive Dirac operator on a half-plane, also studied in \cite{GL}.
Fix some $a>0$. On the $y$-$z$ plane $\RR^2$, consider the operator on $L^2(\RR^2)^{\oplus 2}$ (self-adjoint on the Sobolev space)
\begin{equation*}
H^{\rm 2D}_a:=\begin{pmatrix}-i\partial_z & a-\partial_y \\ a+\partial_y & i\partial_z\end{pmatrix}\; \overset{{\rm Fourier\, in}\; y}{\longleftrightarrow} \; \left\{\begin{pmatrix}-i\frac{d}{dz} & a-ip_y\\ a+ip_y & i\frac{d}{dz}\end{pmatrix}\equiv H^{\rm 1D}(p_x=a,p_y)\right\}_{p_y\in\widehat{\RR}}.
\end{equation*}
Fix some boundary condition $e^{i\gamma}\in{\rm U}(1)$ for the half-plane ($z\geq 0$) operator $\wt{H}^{\rm 2D}_a(\gamma)$, and reparametrise $p_y\in\widehat{\RR}$ in polar coordinates into $\theta=\tan^{-1}(\frac{p_y}{a})\in(-\frac{\pi}{2},\frac{\pi}{2})$, so that the above 1-parameter family of Dirac Hamiltonians becomes the 1-parameter family of massive half-line Dirac Hamiltonians,
\begin{equation*}
\wt{H}^{\rm 1D}_a(\theta;\gamma):\theta\mapsto\begin{pmatrix}-i\frac{d}{dz} & a(1-i\tan\theta) \\ a(1+i\tan\theta) & i\frac{d}{dz}\ \end{pmatrix}\equiv \wt{H}^{\rm 1D}(m=a\sec\theta, \theta;\gamma),
\end{equation*}
which has $(-a\sec\theta, a\sec\theta)$ as a common essential spectral gap. In particular, $H^{\rm 2D}_a$ has this interval as a spectral gap. Utilising the eigenvalue computation of Eq.\ \eqref{eqn:1D.evalue}-\eqref{eqn:1D.efunction}, we plot in Fig.\ \ref{fig:Unusual.flow} the full spectrum of the operator family as a function of the parameter $\theta$ (see also Fig.\ 2 of \cite{GL}). The flow of eigenvalues as $\theta$ is varied depends greatly on the chosen boundary condition $\gamma$. For $\gamma\in (-\pi,0)$, there is spectral flow of $-1$ across zero energy, with a curve of eigenvalues connecting the upper and lower essential spectra (possibly at the end-points of the operator path). Thus $\wt{H}^{\rm 2D}_a(\gamma)$ has the gap-filling phenomenon in this regime. On the other hand, for $\gamma\in (0,\frac{\pi}{2})\cup(\frac{\pi}{2},\pi)$, the curve of eigenvalues only comes out of either the upper or lower essential spectrum, and flows further upwards or downwards to infinity. 

Fourier-transforming $H_a^{\rm 2D}$ in both $y$ and $z$ gives the two-parameter family
\begin{equation*}
\widehat{\RR}^2\ni (p_y,p_z)\mapsto H_a(p_y,p_z):=(a,p_y,p_z)\cdot\vect{\sigma},
\end{equation*}
so there is a negative-energy eigenbundle $\mathcal{E}_-$ over $\widehat{\RR}^2$. The classifying map $\widehat{\RR}^2\rightarrow\CC\PP^1\cong S^2$ for $\mathcal{E}_-$, namely $(a,p_y,p_z)\mapsto \frac{(a,p_y,p_z)}{|(a,p_y,p_z)|}$, covers exactly half of the Bloch sphere $\CC\PP^1\cong S^2$, so it has ``Chern number $\frac{1}{2}$''. Correspondingly, the spectral gap-filling phenomenon occurs for only half of the possible boundary conditions $\gamma$. Thus it is stable in the \emph{weak} sense that small variations of $\gamma$ do not typically destroy the phenomenon. 

Related to the above, we mention that in \cite{GJT}, an interesting model of shallow-water waves was studied in the context of a ``violation of the bulk-edge correspondence'', in the sense that the number of arcs of edge states emerging/entering a bulk spectral band depends not just on the topological invariant of the band (a genuine integer-quantised Chern number) but also the boundary condition, see Fig.\ 1 there.

\end{example}

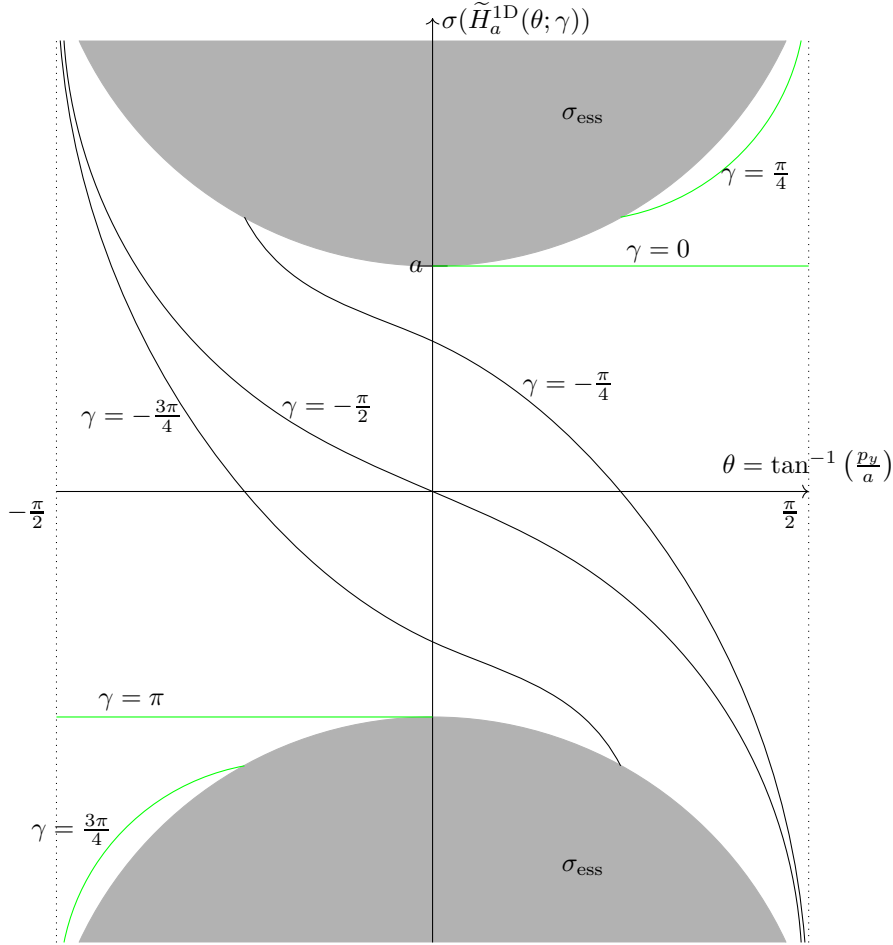
\begin{figure}[h!]
\begin{center}


\begin{tikzpicture}
\draw[dotted] (-5,-6) -- (-5,6);
\draw[dotted] (5,-6) -- (5,6);
\filldraw[gray!60, draw=gray!60] (-4.7,6) to [curve through = {(0,3)}] (4.7,6); 
\filldraw[gray!60, draw=gray!60] (-4.7,-6) to [curve through = {(0,-3)}] (4.7,-6); 
\draw[->] (-5,0) -- (5,0);
\draw[->] (0,-6) -- (0,6.3);
\node[above] at (5,0) {$\theta=\tan^{-1}\left(\frac{p_y}{a}\right)$};
\node[right] at (0,6.3) {$\sigma(\widetilde{H}^{\rm 1D}_a(\theta;\gamma))$};
\node[below left] at (5,0) {$\frac{\pi}{2}$};
\node[below left] at (-5,0) {$-\frac{\pi}{2}$};

\draw (-0.2,3) -- (0.2,3);
\node[above, left] at (0,3) {$a$};
\node at (2,5) {$\sigma_{\rm ess}$};
\node at (2,-5) {$\sigma_{\rm ess}$};

\draw (-2.5,3.65) to [curve through = {(-2,3) (0,2) (2.5,0)}] (4.95,-6);
\draw (2.5,-3.65) to [curve through = {(2,-3) (0,-2) (-2.5,0)}] (-4.95,6);
\draw (-4.9,6) to [curve through = {(-2,1) (0,0) (2,-1)}] (4.9,-6);
\draw[green] (0,3) -- (5,3);
\draw[green] (0,-3) -- (-5,-3);
\draw[green] (2.5,3.65) to [curve through = {(4.5,5)}] (4.9,6);
\draw[green] (-2.5,-3.65) to [curve through = {(-4.5,-5)}] (-4.9,-6);

\node at (4.3,4.2) {$\gamma=\frac{\pi}{4}$};
\node at (3,3.2) {$\gamma=0$};
\node at (1.8,1.4) {$\gamma=-\frac{\pi}{4}$};
\node at (-1.4,1.1) {$\gamma=-\frac{\pi}{2}$};
\node at (-4,1) {$\gamma=-\frac{3\pi}{4}$};
\node at (-4,-2.8) {$\gamma=\pi$};
\node at (-4.8,-4.5) {$\gamma=\frac{3\pi}{4}$};

\end{tikzpicture}

\caption{When the half-plane operator $\wt{H}^{\rm 2D}_a(\gamma)$ of Example \ref{ex:half.Chern} is Fourier transformed in $y$-variable, we get a 1-parameter family of massive half-line Dirac Hamiltonians, $\theta\mapsto \wt{H}^{\rm 1D}_a(\theta;\gamma)$, whose spectrum is plotted above. Spectral flow across the essential spectrum gap only occurs for some choices of boundary condition $\gamma$ (black curves) but not others (green curves).
}\label{fig:Unusual.flow}
\end{center}
\end{figure}

\section{Fermi arcs of Weyl semimetals via spectral flow: tight-binding models}\label{sec:tight.binding}
\subsection{Topological band crossings}\label{sec:band.crossings}
In band theory, one takes a 3D Schr\"{o}dinger operator $H=-\nabla^2+V$ with periodic potential $V$ (i.e.\ invariant only under a lattice $\ZZ^3$ of translations), and carries out Bloch--Floquet transform (Fourier transform with respect to $\ZZ^3$) to produce a family $\{H(\vect{k})\}_{\vect{k}\in\TT^3}$ of Bloch Hamiltonians, parametrised by the Brillouin torus $\TT^3={\rm Hom}(\ZZ^3, {\rm U}(1))$ (``quasi-momentum space''). Effectively, for each $\vect{k}=(k_1,k_2,k_3)\in\TT^3$, one is solving the Schr\"{o}dinger equation on the compact position-space torus $T^3=\RR^3/\ZZ^3$, subject to the quasi-periodicity condition that there is a phase-shift of $e^{ik_j}$ after going around the $j$-th cycle of $T^3$.

Thus each $H(\vect{k})$ is elliptic on $T^3$ with discrete spectrum accumulating at infinity. Under some unrestrictive conditions on $V$, the spectrum of $H(\vect{k})$ varies continuously (or even smoothly, analytically etc.), see \cite{RS4} XIII.16. In any event, we take as an assumption that the spectra of $H(\vect{k})$ as $\vect{k}$ is varied, form smooth bands.

It can happen that a pair of energy bands crosses at some point $\vect{k}=\vect{k}^*$. Near $\vect{k}^*$, in a local trivalisation of the Bloch bundle for the two bands in question, we would write the effective $2\times 2$ Hermitian matrix family
\begin{equation*}
H_{\rm eff}(\vect{k})=a(\vect{k})+\vect{b}(\vect{k})\cdot\vect{\sigma}\label{eqn:tight.binding.vector}
\end{equation*}
for some scalar function $a$ and some 3-component vector field $\vect{b}=(b_1,b_2,b_3)$. We set $a\equiv 0$ for now\footnote{The important condition is really that $H(\vect{k})$ maintains a spectral gap about zero whenever $\vect{k}\neq \vect{k}^*$, or equivalently, $|a(\vect{k})|<|\vect{b}(\vect{k})|$ needs to hold away from $\vect{k}^*$. Given this, the rest of this section can easily accommodate an additional continuous $a$ without modification of the spectral flow values, cf.\ \S\ref{sec:chemical.potential}.}, so that the zeros/singularities of $\vect{b}$ are exactly the band crossing points, such as $\vect{k^*}$. To a first (linear approximation), we have
\begin{equation*}
H_{\rm eff}(\vect{k})\approx H^{\rm lin}_{\rm eff}(\vect{k})=\sum_{i,j=1}^3 b_{ij}(\vect{k}-\vect{k}^*)_j\,\sigma_i,\qquad b_{ij}=\frac{\partial b_i}{\partial k_j}\Big |_{\vect{k}=\vect{k}^*}
\end{equation*}
When ${\rm det\,}(b_{ij})\neq 0$, $\vect{k}^*$ is called a \emph{Weyl point}, and ${\rm sgn}({\rm det\,}(b_{ij}))\in\{-1,+1\}$ is its \emph{chirality}.
As in \S 1.2 of \cite{Witten}, if we modify the metric to $G^{ij}=\sum_{k=1}^3 b^i_{\;k}b^{j}_{\;k}$, corresponding to the transformation $p_i=\sum_{j=1}^3 b_{ij}(k_j-k^*_j)$, we recover $H_{\rm eff}(\vect{p})=\vect{p}\cdot \vect{\sigma}$, which is the Fourier transform of the Weyl Hamiltonian $H^{\rm Weyl}$.

We recognise the chirality of a Weyl point $\vect{k}^*$ as the \emph{local index} of the vector field $\vect{b}$ at the zero $\vect{k}^*$. The chirality is topological in the sense that it remains unchanged even if $\vect{b}$ is modified near $\vect{k}^*$ (without introducing any other zeros), see \S6 of \cite{Milnor}. If the two bands involved are spectrally isolated from the remaining bands and can be trivialised over the \emph{entire} Brillouin torus $\TT^3$, then $\vect{b}$ actually extends to a \emph{global} vector field over $\TT^3$. The sum of local indices over the zeros of $\vect{b}$ is forced to vanish by the Poincar\'{e}--Hopf theorem \cite{MT1, MT}. Consequently, the two bands must cross again at another Weyl point $\vect{k}^\star$ with the opposite chirality, see Fig.\ \ref{fig:Dirac.to.Fermi}.

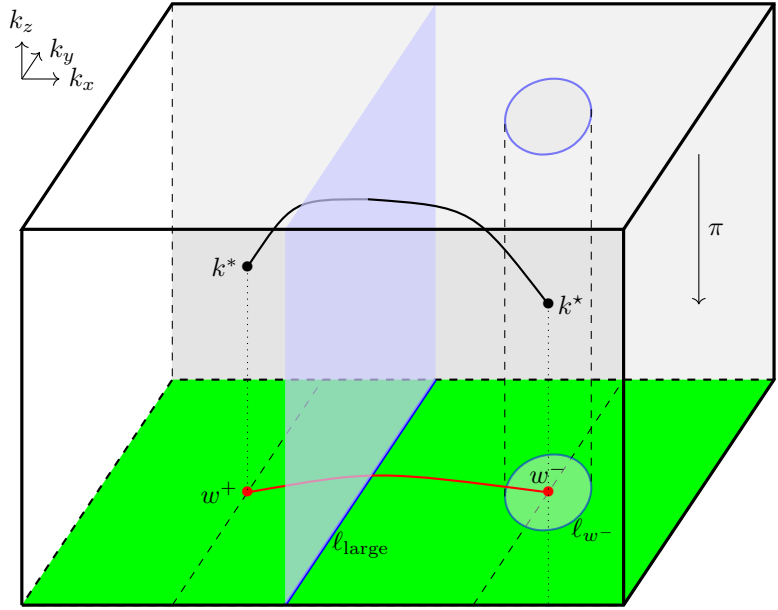
\begin{figure}[h!]
\begin{center}
\begin{tikzpicture}
\fill[gray!30, opacity=0.7] (2,3) -- (8,3) -- (8,5) -- (2,5) -- (2,3);
\fill[gray!10] (2,5) -- (2,8) -- (10,8) -- (10,3) -- (8,0) -- (8,5) -- (2,5);
\filldraw[fill=green, draw=black, dashed, thick] (0,0) -- (8,0) -- (10,3) -- (2,3) -- (0,0);
\draw[dashed] (2,0) -- (4,3);
\draw[dashed] (6,0) -- (8,3);
\draw[dashed] (2,8) -- (2,3);
\draw[->] (9,6) -- (9,4);
\node[right] at (9,5) {$\pi$};
\filldraw[fill=gray!20, opacity=0.5, draw=blue, thick] (7,1) to [curve through = {(7.3,1.1) (7.5,1.8) (7,2) (6.5,1.2) (6.6,1.1)}] (7,1);
\filldraw[fill=gray!20, opacity=0.5, draw=blue, thick] (7,6) to [curve through = {(7.3,6.1) (7.5,6.8) (7,7) (6.5,6.2) (6.6,6.1)}] (7,6);
\node at (7.6,1) {$\ell_{w^-}$};
\draw[dashed](6.42,6.4) -- (6.42,1.4);
\draw[dashed] (7.57,6.6) -- (7.57,1.6);
\draw[very thick, blue] (3.5,0) -- (5.5,3);
\node[right] at (4,0.8) {$\ell_{\rm large}$};
\node[red] at (3,1.5) {$\bullet$};
\node[red] at (7,1.5) {$\bullet$};
\draw[red, thick] (3,1.5) .. controls (4.7,1.8) .. (7,1.5);
\node at (3,4.5) {$\bullet$};
\node[left] at (3,4.5) {$k^*$};
\node[right] at (7,4) {$k^\star$};
\node[left] at (3,1.5) {$w^+$};
\node[above] at (7,1.5) {$w^-$};
\node at (7,4) {$\bullet$};
\draw[dotted] (3,4.5) -- (3,1.5);
\draw[dotted] (7,4) -- (7,0);

\draw[thick] (3,4.5) .. controls (3.6,5.4) .. (4.6,5.4);
\fill[fill=blue!20, opacity=0.7] (3.5,0) -- (5.5,3) -- (5.5,8) -- (3.5,5);
\draw[thick] (4.6,5.4) .. controls (6,5.3) .. (7,4);
\draw[very thick] (0,0) -- (8,0) -- (10,3) -- (10,8) -- (2,8) -- (0,5) -- (8,5) -- (10,8);
\draw[very thick] (0,0) -- (0,5);
\draw[very thick] (8,0) -- (8,5);

\draw[->] (0,7) -- (0.5,7);
\draw[->] (0,7) -- (0,7.5);
\draw[->] (0,7) -- (0.24,7.36);
\node[right] at (0.5,7) {$k_x$};
\node[above] at (0,7.5) {$k_z$};
\node[right] at (0.24,7.36) {$k_y$};
\end{tikzpicture}
\caption{Bulk Brillouin zone $\TT^3$ drawn as a cube with opposite faces identified. The map $\pi:\TT^3\rightarrow \TT^2$  projects onto the surface Brillouin zone $\TT^2$ (green face). For a two-band tight-binding model for a Weyl semimetal Hamiltonian $H_{\rm eff}(\vect{k})=\vect{b}(\vect{k})\cdot\vect{\sigma}$, a pair of Weyl points $\vect{k}^*, \vect{k}^\star$ connected by a Dirac string (black curve) is shown, with projected Weyl points $w^+, w^-$. The spectral flow around a small operator loop $\ell^{\rm op}_{w^\pm}:\xi\mapsto \wt{H}^{\rm 1D}_{\rm eff}(\ell_{w^\pm}(\xi))$ equals the Chern class of the negative eigenbundle $\mathcal{E}_-$ restricted to the cylinder above $\ell_{w^\pm}$. This is $\pm 1$ according to the chirality of the enclosed Weyl point. Thus a Fermi arc must emerge from both $w^\pm$. The connectivity of the arc is deduced by taking the operator loop $\ell^{\rm op}_{\rm  large}:\xi\mapsto \wt{H}^{\rm 1D}_{\rm eff}(\ell_{\rm large}(\xi))$, and deducing that its spectral flow is the Chern class of the negative eigenbundle $\mathcal{E}_-$ restricted to the blue cylinder above $\ell_{\rm large}$. Therefore the Dirac string ``projects'' onto the Fermi arc homologically.}\label{fig:Dirac.to.Fermi}
\end{center}
\end{figure}

\subsection{Tight-binding models of Weyl semimetals}\label{sec:tight.binding.Weyl}
The effective two-band Bloch Hamiltonian family $\{H_{\rm eff}(\vect{k})=\vect{b}(\vect{k})\cdot\vect{\sigma}\}_{\vect{k}\in\TT^3}$ acts on $L^2(\TT^3)^{\oplus 2}$ by multiplication. The spectrum of each $2\times 2$ matrix $H_{\rm eff}(\vect{k})$ is $\pm|\vect{b}(\vect{k})|$, which degenerates to $0$ exactly on the subset $W\subset\TT^3$ of Weyl points. Over $\TT^3\setminus W$, there is a well-defined negative-energy eigenbundle $\mathcal{E}_-$ given by the negative eigenspaces $\mathcal{E}_-(\vect{k})$ of $H_{\rm eff}(\vect{k})$. This line bundle has a Chern class in $H^2(\TT^3\setminus W)$. Observe that $\mathcal{E}_-(\vect{k})$ is just the point in $\CC\PP^1$ corresponding to the unit vector $\hat{\vect{b}}(\vect{k}):=\vect{b}(\vect{k})/|\vect{b}(\vect{k})|\in S^2$ (this is the Bloch sphere identification $S^2\cong\CC\PP^1$, up to a sign choice). Thus, the unit vector field $\vect{b}:\TT^3\setminus W\rightarrow S^2$ can be interpreted as a classifying map for two-band effective gapped Hamiltonians, so that $\mathcal{E}_-$ is the pullback under $\vect{b}$ of the tautological line bundle over $\CC\PP^1$.

\begin{remark}\label{rem:Dirac.string}
We can restrict $\mathcal{E}_-$ and its Chern class to any closed two-submanifold of $\TT^3\setminus W$. For example, if $S^2_{\vect{k}^*}$ is a small 2-sphere enclosing a Weyl point $\vect{k}^*$ but no other points of $W$, then the Chern class of $\mathcal{E}_-|_{S^2_{\vect{k}^*}}$ in $H_{\rm eff}^2(S^2_{\vect{k}^*})\cong\ZZ$ is $\pm 1$ according to the chirality of $\vect{k}^*$. In this sense, each Weyl point serves as a \emph{Dirac monopole} for the ${\rm U}(1)$ line bundle $\mathcal{E}_-$. As explained in \cite{MT1}, the Poincar\'{e} dual to the Chern class of $\mathcal{E}_-$ is its \emph{Dirac string} --- a relative homology class in $H_1(\TT^3,W)$.
\end{remark}

\subsubsection{Two-band Weyl semimetal on the half-space}
If we only Fourier transform in $x,y$, we can rewrite the $H_{\rm eff}$ as 
\begin{equation*}
H_{\rm eff}\cong \int^\oplus_{\vect{k}\in\TT^3} H_{\rm eff}(\vect{k})=\int^\oplus_{(k_x,k_y)\in\TT^2}\underbrace{H_{\rm eff}^{\rm 1D}(k_x,k_y)}_{\int^{\oplus}_{k_z\in\TT} H_{\rm eff}(k_x,k_y,k_z)},
\end{equation*}
where each $H_{\rm eff}^{\rm 1D}(k_x,k_y)$ acts on $L^2(\TT)^{\oplus 2}\cong \ell^2(\ZZ)^{\oplus 2}$ and has spectrum 
\begin{equation*}
\sigma(H_{\rm eff}^{\rm 1D}(k_x,k_y))=\cup_{k_z\in\TT}\;\sigma(H_{\rm eff}(k_x,k_y,k_z)).
\end{equation*}
Since each $H_{\rm eff}^{\rm 1D}(k_x,k_y)$ is itself decomposable into a continuous family of $2\times 2$ matrices $\{H_{\rm eff}(k_x,k_y,k_z)\}_{k_z\in\TT}$, it may be considered as a self-adjoint element in the $C^*$-algebra $M_2(C(\TT^2))$. The eigenvalues of $H_{\rm eff}(k_x,k_y,k_z)$ vary continuously in $k_z$, so $\sigma(H_{\rm eff}^{\rm 1D}(k_x,k_y))$, being the union over $k_z\in\TT$ of $\sigma(H_{\rm eff}(k_x,k_y,k_z))$, comprises only essential spectrum. 

Let $\tilde{W}\subset\TT^2$ denote the set of projected Weyl points, i.e.\ the image of $W$ under the map $\pi:\TT^3\rightarrow\TT^2$ projecting out the $k_z$ variable. Then away from $\tilde{W}$, we have a spectral gap,
\begin{equation*}
0\not\in\sigma_{(\rm ess)}(H_{\rm eff}^{\rm 1D}(k_x,k_y))\;\Leftrightarrow\; (k_x,k_y)\in\TT^2\setminus\tilde{W}.
\end{equation*}

Our real interest is in the upper half-space operator $\wt{H}_{\rm eff}$. To define this, we factorise $L^2(\TT^3)$ as $L^2(\TT^2)\otimes L^2(\TT)\cong L^2(\TT^2)\otimes \ell^2(\ZZ)$. Let $L^2(\TT^2)\otimes \ell^2(\NN)$ be the ``upper-half'' Hilbert subspace, and denote its inclusion and projection maps by $\iota$ and $p$ respectively. Then $\wt{H}_{\rm eff}$ is the compression $p\circ H_{\rm eff}\circ \iota$ to an operator on $(L^2(\TT^2)\otimes \ell^2(\NN))^{\oplus 2}$. We have the compressed version of the Fourier decomposition,
\begin{equation*}
\wt{H}_{\rm eff}\cong \int^\oplus_{(k_x,k_y)\in\TT^2} \wt{H}_{\rm eff}^{\rm 1D}(k_x,k_y)
\end{equation*}
where each $\wt{H}_{\rm eff}^{\rm 1D}(k_x,k_y)$ is the compression of $H_{\rm eff}^{\rm 1D} (k_x,k_y)$ to the ``upper-half-line'' subspace $\ell^2(\NN)^{\oplus 2}$. 

{\bf Toeplitz algebra.}
The Fourier transform identifies $\ell^2(\NN)\subset \ell^2(\ZZ)$ with the Hardy subspace of $L^2(\TT)$. This means that $\wt{H}_{\rm eff}^{\rm 1D}(k_x,k_y)$ gives an element of $M_2(\mathcal{T})$, where $\mathcal{T}$ is the \emph{Toeplitz algebra}, and $M_2(\cdot)$ denotes the $2\times 2$ matrix algebra over $(\cdot)$, see \S 1.3 of \cite{PSB} for a further discussion of 1D tight-binding Hamiltonians and $\mathcal{T}$. 

With $\mathcal{K}$ denoting the compact operators on $\ell^2(\NN)$, the fundamental Toeplitz $C^*$-algebra exact sequence is
\begin{equation}
0\rightarrow \mathcal{K}\rightarrow\mathcal{T}\overset{q}{\rightarrow} C(\TT)\rightarrow 0,\label{eqn:basic.Toeplitz}
\end{equation}
which extends to matrix algebras $M_n(\cdot)$ in the obvious way. The quotient map (or \emph{symbol map}) $q$ undoes the compression $p\circ(\cdot)\circ\iota$, taking $\wt{H}_{\rm eff}^{\rm 1D}(k_x,k_y)\mapsto H_{\rm eff}^{\rm 1D}(k_x,k_y)$. The essential spectrum of $\wt{H}_{\rm eff}^{\rm 1D}(k_x,k_y)$ is its spectrum modulo compacts, so
\begin{equation*}
\sigma_{\rm ess}(\wt{H}_{\rm eff}^{\rm 1D}(k_x,k_y))=\sigma_{\rm ess} (H_{\rm eff}^{\rm 1D}(k_x,k_y)).
\end{equation*}
Notice that modifying $\wt{H}_{\rm eff}^{\rm 1D}(k_x,k_y)$ by a compact perturbation keeps its $\sigma_{\rm ess}$ intact --- such a perturbation could be thought of as a change of boundary conditions in the tight-binding setting.

\subsection{$K$-theory computation of spectral flow in two-band Weyl semimetals}\label{sec:tight.binding.sf}
\emph{Individually}, for each $(k_x,k_y)\in\TT^2$, the operator $\wt{H}_{\rm eff}^{\rm 1D}(k_x,k_y)$ has some extra discrete spectrum (``edge states''), which is \emph{not} stable against compact perturbations. \emph{Collectively}, however, any continuous \emph{loop} $\ell:S^1\rightarrow\TT^2\setminus\tilde{W}$, with coordinates $\ell(\xi)=(k_x(\xi),k_y(\xi))$, gives a corresponding norm-continuous operator loop 
\begin{equation*}
\ell^{\rm op}:S^1\ni \xi\mapsto \wt{H}_{\rm eff}^{\rm 1D}(\ell(\xi))\equiv\wt{H}_{\rm eff}^{\rm 1D}(k_x(\xi),k_y(\xi))
\end{equation*}
in the bounded self-adjoint Fredholms $\mathcal{F}^{\rm sa}_*$. So $\ell^{\rm op}$ has a spectral flow ${\rm sf}(\ell^{\rm op})\in\ZZ$, dependent only on the homotopy class $[\ell]\in\pi_1(\TT^2\setminus\tilde{W})$, due to Eq. \eqref{eqn:sf.commuting.diagram}. It actually suffices to take $[\ell]$ in the abelianisation $H_1(\TT^2\setminus\tilde{W})$. A compact perturbation $K(\xi)$, depending continuously on $\xi$, can even be added to such an operator loop without changing its spectral flow. 

Because the half-line operator family $\{\wt{H}_{\rm eff}^{\rm 1D}(k_x,k_y)\}_{(k_x,k_y)\in\TT^2}$ can be quite complicated, it is not efficient, nor generally possible, to compute the spectral flow of loops $\ell^{\rm op}$ in this family by solving the spectral problems directly. Instead, we will explain how to compute the spectral flow structure using $K$-theory.

Each operator loop $\ell^{\rm op}$ above is a continuous loop of Toeplitz operators, defining an element of the algebra $C(S^1)\otimes M_2(\mathcal{T})$. The latter algebra sits (after taking $M_2(\cdot)$) in the middle of the Toeplitz exact sequence (Eq. \eqref{eqn:basic.Toeplitz}) tensored with $C(S^1)$:
\begin{equation}
0\rightarrow C(S^1)\otimes \mathcal{K}\rightarrow C(S^1)\otimes \mathcal{T}\overset{{\rm id}\otimes q}{\longrightarrow} C(S^1)\otimes C(\TT)\rightarrow 0.\label{eqn:SES}
\end{equation}
The quotient map ${\rm id}\otimes q$ converts the half-line operator loop $\ell^{\rm op}$ back to the full-line operator loop 
\begin{equation*}
\breve{\ell}^{\rm op}:S^1\ni \xi\mapsto H^{\rm 1D}_{\rm eff}(\ell(\xi)) \cong \int^\oplus_{k_z\in \TT} \underbrace{H_{\rm eff}(\ell(\xi), k_z)}_{H_{\rm eff}(k_x(\xi), k_y(\xi), k_z)}.
\end{equation*}
We may Fourier transform $\breve{\ell}^{\rm op}$ in the $z$-variable, obtaining a 2-torus-parametrised family of $2\times 2$ matrices
\begin{equation*}
S^1\times\TT\ni (\xi,k_z) \mapsto H_{\rm eff}(\ell(\xi), k_z)\equiv H_{\rm eff}(k_x(\xi), k_y(\xi), k_z)\equiv\vect{b}(\ell(\xi), k_z)\cdot\vect{\sigma}.
\end{equation*}
We shall refer to $S^1\times \TT$ as a ``cylinder'' $S^1\times [0,2\pi]$, with the understanding that the upper and lower boundary circles are identified. Notice that the entire cylinder avoids the Weyl points $W\subset\pi^{-1}(\tilde{W})$, so that all the matrices $H_{\rm eff}(\ell(\xi), k_z)$ are spectrally-gapped at $0$. Thus we can extract out the negative eigenbundle over the cylinder, using the continuous projection family
\begin{align}
P_{-,\ell}&=\left\{(\xi,k_z)\mapsto\frac{1-{\rm sgn}(H_{\rm eff}(\ell(\xi), k_z))}{2}\right\}_{(\xi,k_z)\in S^1\times\TT}\label{eqn:loop.of.projections}\\
&=\left\{(\xi,k_z)\mapsto\frac{1-{\rm sgn}(\vect{b}(\ell(\xi))\cdot\vect{\sigma})}{2}\right\}_{(\xi,k_z)\in S^1\times\TT}\quad\in M_2(C(S^1)\otimes C(\TT)).\nonumber
\end{align}
This projection gives a $K$-theory class $[P_{-,\ell}]\in K_0(C(S^1)\otimes C(\TT))$; in terms of topological $K$-theory, this is the class of $\mathcal{E}_-$ in $K^0(S^1\times\TT)$. In fact, via the classifying vector field $\vect{b}$, the latter $K$-theory class is just pulled back from that of the tautological (Hopf) bundle over $\CC\PP^1$ to the cylinder.

In the $K$-theory long exact sequence for Eq. \eqref{eqn:SES}, the connecting exponential map
\begin{equation*}
{\rm Exp}:\underbrace{K_0(C(S^1)\otimes C(\TT))}_{\ZZ\oplus\ZZ}\rightarrow \underbrace{K_1(C(S^1)\otimes \mathcal{K})}_{K_1(C(S^1))}\cong\ZZ\label{eqn:K.exponential}
\end{equation*}
takes the non-trivial (Hopf) generator to a generator. We refer to \S2 of \cite{Thiang} for more details, and only report what ${\rm Exp}$ does. The group $K_0(C(S^1)\otimes C(\TT))\cong\ZZ\oplus\ZZ$ has one ``trivial'' generator represented by the identity projection ${\mathbf 1}_{C(S^1)\otimes C(\TT)}$, while the second generator $[P_{\rm Chern}]$ comes from pulling back the Hopf projection in $M_2(C(S^2))$ under a degree $1$ map\footnote{The 2D Chern insulator has, by definition, its negative eigenprojection in the class of $P_{\rm Chern}$, thus the name. There are sign conventions at various stages, such how $\CC\PP^1$ is identified with the unit vectors of $S^2$ in the Bloch sphere construction, and also the Chern number $\pm 1$ of the tautological (Hopf) bundle.} $S^1\times\TT\rightarrow \CC\PP^1\cong S^2$. Also, $K_1(C(S^1))\cong\ZZ$ has generator represented by the unitary map ${\mathtt w}:e^{i\xi}\mapsto e^{i\xi}$ with winding number 1. Then
\begin{equation*}
{\rm Exp}:[P_{\rm Chern}]\mapsto -[{\mathtt w}],
\end{equation*}
and this non-trivial map plays a fundamental role in proving the \emph{gap-filling property} for half-space Chern insulators \cite{PSB,Thiang}. 
The gap-filling property actually persists for more complicated half-space truncations than $z\geq 0$ \cite{Thiang,LT}. 

Let us compute the $K$-theory class ${\rm Exp}[P_{-,\ell}]$ in terms of the spectral flow of $\ell^{\rm op}$. Recall from Eq.\ \eqref{eqn:loop.of.projections} that $P_{-,\ell}$ is a loop of projection operators on $\ell^2(\ZZ)^{\oplus 2}$. 
By construction, the spectral-flattening operation $\phi$ (via functional calculus, recall \S{\ref{sec:spectral.flow.generalities})) has the property that 
\begin{equation*}
q\left(\phi\left(\wt{H}^{\rm 1D}_{\rm eff}(\ell(\xi))\right)\right)=\phi\left(q\left(\wt{H}^{\rm 1D}_{\rm eff}(\ell(\xi))\right)\right)=\phi\left(H^{\rm 1D}_{\rm eff}(\ell(\xi))\right)={\rm sgn}\left(H^{\rm 1D}_{\rm eff}(\ell(\xi))\right).
\end{equation*}
Therefore, the projection loop $P_{-,\ell}$ lifts (under ${\rm id}\otimes q$) to the self-adjoint loop
\begin{equation*}
Q_{-,\ell}:=\left\{\xi\mapsto \frac{1-\phi\left(\wt{H}^{\rm 1D}_{\rm eff}(\ell(\xi))\right)}{2}\right\}_{\xi\in S^1}.
\end{equation*}
The lifted operators $Q_{-,\ell}(\xi)$ are not necessarily projections, since $\phi\left(\wt{H}^{\rm 1D}_{\rm eff}(\ell(\xi))\right)$ may acquire spectra outside of $\{-1,1\}$. In fact, an obstruction to finding a \emph{projection} lift of $P_{-\ell}$, is precisely given by the $K$-theory connecting map ${\rm Exp}$.

By definition of ${\rm Exp}:K_0(C(S^1)\otimes C(\TT))\rightarrow K_1(C(S^1))$, 
\begin{align*}
{\rm Exp}[P_{-,\ell}]:= [{\rm exp}(-2\pi i Q_{-,\ell})]&=\left[\xi\mapsto {\rm exp}\left(i\pi \left(\phi(\wt{H}^{\rm 1D}_{\rm eff}\left(\ell(\xi)\right))+1\right)\right)\right]\\
&\in [S^1,{\rm U}(\infty)]= K_1(C(S^1)).
\end{align*}
Applying the winding homomorphism and Eq. \eqref{eqn:sf.commuting.diagram}, this gives
\begin{equation}
{\rm Wind}\circ{\rm Exp}[P_{-,\ell}]={\rm sf}(\wt{H}^{\rm 1D}_{\rm eff}(\ell(\cdot)))\equiv {\rm sf}(\ell^{\rm op}).\label{eqn:exponential.sf}
\end{equation}

Of special interest are the small loops $\ell_w$ encircling a projected Weyl point $w\in\tilde{W}$, see Fig.\ \ref{fig:Dirac.to.Fermi}. 
If $w$ is the projected Weyl point for a single Weyl point of chirality $+1$, the classifying vector field $\vect{b}$ restricts to a degree $1$ map from the cylinder $S^1\times \TT$ to $S^2\cong\CC\PP^1$, so that $[P_{-,\ell_w}]\in K_0(C(S^1)\otimes C(\TT))$ is the non-trivial (Hopf) generator $[P_{\rm Chern}]$. Then Eq.\ \eqref{eqn:exponential.sf} is simply
\begin{equation*}
-1 ={\rm Wind}(-[\mathtt w])={\rm Wind}({\rm Exp}[P_{-,\ell_w}]) = {\rm sf}(\ell^{\rm op}_w).\label{eqn:sf.around.WP}
\end{equation*}
If $w$ is the projection image of several different Weyl points, then the degree of the cylinder-restricted $\vect{b}$, and thus the $K$-theory class of $P_{-,\ell_w}$, is the sum of the local indices of $\vect{b}$ at the enclosed Weyl points. Thus the spectral flow of $\ell^{\rm op}_w$ is the sum of these indices by the additivity of spectral flow.

Also of interest are the ``large loops'' $\ell_{\rm large}$ in $\TT^2\setminus \tilde{W}$, wrapping around a cycle of $\TT^2$, see Fig.\ \ref{fig:Dirac.to.Fermi}. Sitting above $\ell_{\rm large}$ is a large cylinder (actually a 2-torus slice) in $\TT^3$. The degree of the classifying vector field $\vect{b}$ restricted to such a large cylinder is \emph{not} determined just by the local indices of the Weyl points, see \cite{MT, MT1}, but whatever the degree is (sometimes referred to as ``2D Chern numbers''), it will determine the spectral flow of $\ell^{\rm op}_{\rm large}$. Globally, all these degrees, and therefore the resultant spectral flows, are encoded by counting intersections with the Dirac strings (Remark \ref{rem:Dirac.string}) connecting the Weyl points.

\subsection{Homology of Fermi arcs}\label{sec:homology.Fermi.arc}
By considering all the different homology classes of loops $\ell$ in $\TT^2\setminus \tilde{W}$ in turn, we can map out the spectral flow structure of the Weyl semimetal $(k_x,k_y)\mapsto \wt{H}^{\rm 1D}_{\rm eff}(k_x,k_y)$ completely. For a loop $\ell^{\rm op}$ with non-trivial spectral flow $n$, there are $n$ (nett) points on the loop where $\ell^{\rm op}(\xi)$ has $0$-eigenvalues, and these points contribute to the Fermi arc of the Weyl semimetal. By Definition \ref{defn:Fermi.arc}, the Fermi arcs are the intersection of the spectral graph $y(k_x,k_y)=\sigma(\wt{H}^{\rm 1D}_{\rm eff}(k_x,k_y))$, with the zero energy level surface $y=0$, and represent a class in $H_1(\TT^2,\tilde{W})$. The latter homology class is completely specified by its intersection numbers with $\ell$ (Poincar\'{e} duality), in other words, by the various spectral flows ${\rm sf}(\ell^{\rm op})$. In turn, the latter spectral flows are fully determined by the Dirac strings of $H_{\rm eff}$.
The upshot is that the Fermi arcs represent the class in $H_1(\TT^2,\tilde{W})$ obtained by homological projection (under the map $\TT^3\rightarrow \TT^2$) of the Dirac strings of the Weyl semimetal $H_{\rm eff}$ --- the thesis of \cite{MT1}, summarised in Fig.\ \ref{fig:Dirac.to.Fermi}.

The true effective half-space Hamiltonian $\wt{H}_{\rm eff,\, true}$ may actually be the compression $\wt{H}_{\rm eff}$ plus some boundary perturbation term $M_2(C(S^1)\otimes\mathcal{K})$, in accordance with Eq.\ \eqref{eqn:SES}. Nevertheless, besides Weyl-point-preserving deformations of $\vect{b}$, such perturbations do not change the above $K$-theory computations of the spectral flow structure for $\wt{H}_{\rm eff,\, true}$. Therefore, \emph{the homology class of the resultant Fermi arcs are perturbation-resistant}, even though the perturbation can effect dramatic geometric deformations of the Fermi arcs, possibly even ``rewiring'' them (see \cite{TSG} for examples of such rewirings).

\section{Topologically connected Fermi arcs via spectral flow: continuum models}\label{sec:top.analysis}
As in \S\ref{sec:tight.binding.Weyl}, consider a Weyl semimetal with a pair $\vect{k}^*, \vect{k}^\star$ of \emph{Weyl points}, which without loss, we assume to have chirality $+1$ and $-1$ respectively. Locally in $\vect{k}$-space, they are approximated by a right and left-handed Weyl Hamiltonian, respectively, up to a linear change of metric. However, it is very important to also have the \emph{global} information of how the Weyl points are connected in the Brillouin torus $\TT^3$, as emphasised theoretically in \cite{MT, MT1}. Experimentally, Fermi arcs have been observed to traverse Brillouin zone ``boundaries'' \cite{Souma,Morali}. At the level of \emph{topology}, the connectivity of Weyl points, as recordable by Dirac strings in the 3D Brillouin zone $\TT^3$, is projected onto the Fermi arc connectivity in $\TT^2$, as explained in \S\ref{sec:homology.Fermi.arc}.

In the effective two-band tight-binding Hamiltonian model $H_{\rm eff}$, the role of boundary conditions in determining the \emph{geometry} of Fermi arcs is less clear. This is because the boundary conditions are also enforced on the half-space operator $\wt{H}_{\rm eff}$ only as an effective perturbation. In order to understand phenomena like spin-momentum locking, we need a differential operator description, as we saw in the half-space Weyl Hamiltonian $\wt{H}^{\rm Weyl}(\gamma)$ example in \S\ref{sec:Weyl.Fermi}. We cannot, however, hope to deduce that a Fermi arc connects a pair of projected Weyl points, simply by using independent local Weyl Hamiltonian approximations valid only near each Weyl point. The local description needs to contain \emph{both} Weyl points.

\subsection{Generic differential operator model for a pair of connected Weyl points}\label{sec:generic.model}
We will study a fairly general class of 3D self-adjoint constant coefficient differential operators $H^{\rm Weyls}$ on $(L^2(\RR^3))^{\oplus 2}$ which are \emph{second-order} in the directions parallel to the boundary $z=0$, but first-order in the normal direction $z$. This improves some basic models studied in \cite{Witten, BLP}. Explicitly, the symbol (Fourier transform) of $H^{\rm Weyls}$ is taken to be of the form
\begin{equation}
H^{\rm Weyls}(\vect{p})=\begin{pmatrix} p_z &  \overline{g(p_x,p_y)} \\ g(p_x,p_y) & -p_z \end{pmatrix},\qquad \vect{p}=(p_x,p_y,p_z)\in\widehat{\RR}^3, \label{eqn:Weyls.symbol}
\end{equation}
where $g:\widehat{\RR}^2\rightarrow \CC$ is some degree-$2$ polynomial\footnote{Observe that $g={\rm id}$ gives the right-handed Weyl Hamiltonian, while $g(p_x,p_y)=p_x-ip_y$ gives some spin-rotated version of the left-handed Weyl Hamiltonian.}. Here we have switched the notation from $\vect{k}$ to $\vect{p}$ to emphasise that we are no longer taking the momentum modulo $2\pi$. 

The spectrum of $H^{\rm Weyls}(\vect{p})$ degenerates to zero exactly at the points $\vect{p}=(p_x,p_y,p_z=0)$ such that $g(p_x,p_y)=0$. While this restricts Weyl points to always have $p_z=0$, a suitable gauge transformation of the model Eq.\ \eqref{eqn:Weyls.symbol} can place the Weyl points at any general $p_z$, see \S\ref{sec:shifted.Weyl.points} for details.

As we will shortly introduce the boundary $z=0$, we only keep the Fourier transform $H^{\rm Weyls}$ in the $x,y$ directions. Thus Eq.\ \eqref{eqn:Weyls.symbol} is rewritten as a 2-parameter family
\begin{equation}
H^{\rm 1D}_{\rm Weyls}(p_x,p_y)=\begin{pmatrix} -i\frac{d}{dz} &  \overline{g(p_x,p_y)} \\ g(p_x,p_y) & i\frac{d}{dz} \end{pmatrix}\equiv H^{\rm 1D}(g(p_x,p_y)), \qquad (p_x,p_y)\in\widehat{\RR}^2.\label{eqn:Weyls.decomposition}
\end{equation}
of 1D Dirac Hamiltonians, which becomes massless precisely at the zeros of $g$.

It is convenient to regard $g:\widehat{\RR}^2\rightarrow\CC\cong\widehat{\RR}^2$ as a two-component vector field on the surface momentum space $\widehat{\RR}^2$. So each zero of $g$ has an associated local index, given by the winding number of $g$ taken along a small loop encircling the zero. Note that $\widehat{\RR}^2$ is not compact, so there is no global index cancellation unless some conditions on $g$ at infinity are given.

To produce a pair of (distinct) Weyl points $\vect{p}^\pm=(p^\pm_x,p^\pm_y,0)$ with respective chiralities $\pm 1$, we need the vector field $g$ to have zeros at $w^\pm=(p^\pm_x,p^\pm_y)$. Writing $w=(p_x,p_y)=p_x+ip_y$ as a complex variable, the polynomial
\begin{equation}
g(p_x,p_y)\equiv g(w)=(w-w^+)\overline{(w-w^-)}\label{eqn:model.vector.field}
\end{equation}
works, and has the added bonus that as $|w|^2=p_x^2+p_y^2\rightarrow\infty$, the vector field $g$ has an isotropic limiting direction (pointing to the right), see Fig.\ \ref{fig:Weyl.vector.field} for a sketch. This means that $g$ does not wind around the point at infinity. We remark that the simpler models in \cite{Witten, BLP} do not have this feature, and this has consequences for the global connectivity of Fermi arcs. 

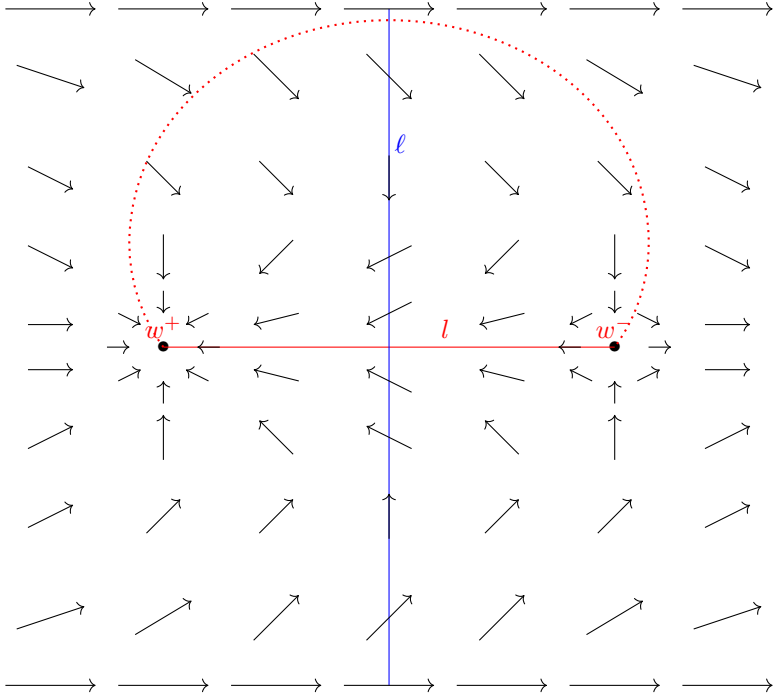
\begin{figure}[h!]
\begin{center}

\begin{tikzpicture}[scale=1.5]
\node at (-2,0) {$\bullet$};
\node at (2,0) {$\bullet$};
\node[red, above] at (-2,0) {$w^+$};
\node[red, above] at (2,0) {$w^-$};
\draw[red] (-2,0) -- (2,0);
\node[red, above] at (0.5,0) {$l$};
\draw[blue] (0,-3) -- (0,3);
\node[blue] at (0.1,1.8) {$\ell$}; 

\draw[red, thick, dotted] (-2,0) to [curve through = {(-2.3,1)(0,2.9) (2.3,1)}] (2,0);

\draw[->] (-3.4,3) -- (-2.6,3);
\draw[->] (-2.4,3) -- (-1.6,3);
\draw[->] (-1.4,3) -- (-0.6,3);
\draw[->] (-0.4,3) -- (0.4,3);
\draw[->] (0.6,3) -- (1.4,3);
\draw[->] (1.6,3) -- (2.4,3);
\draw[->] (2.6,3) -- (3.4,3);

\draw[->] (-3.4,-3) -- (-2.6,-3);
\draw[->] (-2.4,-3) -- (-1.6,-3);
\draw[->] (-1.4,-3) -- (-0.6,-3);
\draw[->] (-0.4,-3) -- (0.4,-3);
\draw[->] (0.6,-3) -- (1.4,-3);
\draw[->] (1.6,-3) -- (2.4,-3);
\draw[->] (2.6,-3) -- (3.4,-3);

\draw[->] (-3.3,2.5) -- (-2.7,2.3);
\draw[->] (-2.25,2.55) -- (-1.75,2.25);
\draw[->] (-1.2,2.6) -- (-0.8,2.2);
\draw[->] (-0.2,2.6) -- (0.2,2.2);
\draw[->] (0.8,2.6) -- (1.2,2.2);
\draw[->] (1.75,2.55) -- (2.25,2.25);
\draw[->] (2.7,2.5) -- (3.3,2.3);

\draw[->] (-3.3,-2.5) -- (-2.7,-2.3);
\draw[->] (-2.25,-2.55) -- (-1.75,-2.25);
\draw[->] (-1.2,-2.6) -- (-0.8,-2.2);
\draw[->] (-0.2,-2.6) -- (0.2,-2.2);
\draw[->] (0.8,-2.6) -- (1.2,-2.2);
\draw[->] (1.75,-2.55) -- (2.25,-2.25);
\draw[->] (2.7,-2.5) -- (3.3,-2.3);

\draw[->] (-3.2,1.6) -- (-2.8,1.4);
\draw[->] (-2.15,1.65) -- (-1.85,1.35);
\draw[->] (-1.15,1.65) -- (-0.85,1.35);
\draw[->] (0,1.7) -- (0,1.3);
\draw[->] (0.85,1.65) -- (1.15,1.35);
\draw[->] (1.85,1.65) -- (2.15,1.35);
\draw[->] (2.8,1.6) -- (3.2,1.4);

\draw[->] (-3.2,-1.6) -- (-2.8,-1.4);
\draw[->] (-2.15,-1.65) -- (-1.85,-1.35);
\draw[->] (-1.15,-1.65) -- (-0.85,-1.35);
\draw[->] (0,-1.7) -- (0,-1.3);
\draw[->] (0.85,-1.65) -- (1.15,-1.35);
\draw[->] (1.85,-1.65) -- (2.15,-1.35);
\draw[->] (2.8,-1.6) -- (3.2,-1.4);

\draw[->] (-3.2,0.9) -- (-2.8,0.7);
\draw[->] (-2,1) -- (-2,0.6);
\draw[->] (-0.85,0.95) -- (-1.15,0.65);
\draw[->] (0.2,0.9) -- (-0.2,0.7);
\draw[->] (1.15,0.95) -- (0.85,0.65);
\draw[->] (2,1) -- (2,0.6);
\draw[->] (2.8,0.9) -- (3.2,0.7);

\draw[->] (-3.2,-0.9) -- (-2.8,-0.7);
\draw[->] (-2,-1) -- (-2,-0.6);
\draw[->] (-0.85,-0.95) -- (-1.15,-0.65);
\draw[->] (0.2,-0.9) -- (-0.2,-0.7);
\draw[->] (1.15,-0.95) -- (0.85,-0.65);
\draw[->] (2,-1) -- (2,-0.6);
\draw[->] (2.8,-0.9) -- (3.2,-0.7);

\draw[->] (-2,0.5) -- (-2,0.3);
\draw[->] (-2.4,0.3) -- (-2.2,0.2);
\draw[->] (-2.5,0) -- (-2.3,0);
\draw[->] (-2.4,-0.3) -- (-2.2,-0.2);
\draw[->] (-2,-0.5) -- (-2,-0.3);
\draw[->] (-1.6,-0.3) -- (-1.8,-0.2);
\draw[->] (-1.5,0) -- (-1.7,0);
\draw[->] (-1.6,0.3) -- (-1.8,0.2);

\draw[->] (2,0.5) -- (2,0.3);
\draw[->] (2.2,0.3) -- (2.4,0.2);
\draw[->] (2.3,0) -- (2.5,0);
\draw[->] (2.2,-0.3) -- (2.4,-0.2);
\draw[->] (2,-0.5) -- (2,-0.3);
\draw[->] (1.8,-0.3) -- (1.6,-0.2);
\draw[->] (1.7,0) -- (1.5,0);
\draw[->] (1.8,0.3) -- (1.6,0.2);

\draw[->] (-3.2,0.2) -- (-2.8,0.2);
\draw[->] (-3.2,-0.2) -- (-2.8,-0.2);
\draw[->] (2.8,0.2) -- (3.2,0.2);
\draw[->] (2.8,-0.2) -- (3.2,-0.2);

\draw[->] (-0.8,0.3) -- (-1.2,0.2);
\draw[->] (-0.8,-0.3) -- (-1.2,-0.2);
\draw[->] (1.2,0.3) -- (0.8,0.2);
\draw[->] (1.2,-0.3) -- (0.8,-0.2);
\draw[->] (0.2,0.4) -- (-0.2,0.2);
\draw[->] (0.2,-0.4) -- (-0.2,-0.2);

\end{tikzpicture}
\caption{An example of a vector field $g:\widehat{\RR}^2\rightarrow \widehat{\RR}^2\subset\mathcal{M}$, with zeros of opposite indices at distinct points $w^+, w^-$. This gives a field of half-line Dirac Hamiltonians $\{H^{\rm 1D}_{\rm Weyls}(p_x,p_y)\}_{(p_x,p_y)\in\widehat{\RR}^2}$ pulled back from $\mathcal{M}$, which assemble into the 3D half-space differential operator Eq.\ \eqref{eqn:Weyls.symbol} modelling a pair of connected Weyl points. There is spectral flow along loops/paths for which $g$ rotates a non-zero number of times, for example, the small loops encircling $w^\pm$, or the blue line $\ell$. These windings, thus spectral flows, determine the topology of the Fermi arc (represented by the solid red line), but depending on the boundary conditions, the actual Fermi arc will be some deformed version (such as the dotted red line).
}\label{fig:Weyl.vector.field}
\end{center}
\end{figure}

\subsubsection{Spectral flow structure for $\wt{H}^{\rm Weyls}$}
Analogously to Eq.\ \eqref{eqn:Weyls.decomposition}, for the half-space version of $H^{\rm Weyls}$, we get a decomposition into a family of \emph{half-line} Dirac Hamiltonians,
\begin{align*}
 \wt{H}^{\rm Weyls}(\gamma)&\rightsquigarrow \wt{H}^{\rm 1D}_{\rm Weyls}(p_x,p_y;\gamma(p_x,p_y))\\
 &=\wt{H}^{\rm 1D}\left(g(p_x,p_y);\gamma(p_x,p_y)\right)\in\mathcal{M},\qquad(p_x,p_y)\in\widehat{\RR}^2,
\end{align*}
according to some boundary condition function $\gamma=\gamma(p_x,p_y)$. Thus $\wt{H}^{\rm Weyls}(\gamma)$ is really given by a ``classifying map'' 
\begin{equation*}
\wt{H}^{\rm Weyls}(\gamma)\;\;\longleftrightarrow\;\; \wt{h}^{\rm Weyls;\gamma}:\widehat{\RR}^2\rightarrow \mathcal{M},
\end{equation*}
whose restriction to $\widehat{\RR}^2\setminus\{w^+,w^-\}$ lands in $\mathcal{M}^\prime$. Since $\gamma$ is defined over the entire $\widehat{\RR}^2$, for the purposes of computing spectral flow of this family (of massive half-line Dirac Hamiltonians), we can assume by Cor.\ \ref{cor:general.Dirac.loop.sf} that $\gamma$ is just some constant (cf.\ Prop.\ \ref{prop:boundary.independence}). Then $\wt{h}^{\rm Weyls;\gamma}$ factorises as
\begin{equation}
\wt{h}^{\rm Weyls;\gamma}=(\wt{h}^{\rm Weyl;\gamma}\circ g)\;:\;\widehat{\RR}^2\rightarrow \mathcal{M},\label{eqn:factorisation}
\end{equation}
where $\wt{h}^{\rm Weyl;\gamma}:\widehat{\RR}^2\rightarrow\mathcal{M}$ is the canonical classifying map (see Eq.\ \eqref{eqn:Weyl.function}) for the standard half-space Weyl Hamiltonian $\wt{H}^{\rm Weyl}(\gamma)$. In particular, $\wt{h}^{\rm Weyls;\gamma}$ is gap-continuous.

Pick any continuous momentum space loop $\ell:S^1\rightarrow \wh{\RR}^2\setminus\{w^+,w^-\}$ that avoids the projected Weyl points --- this determines a gap-continuous operator loop 
\begin{equation*}
\ell^{\rm op}=\wt{h}^{\rm Weyls;\gamma}\circ\ell=\wt{h}^{\rm Weyl;\gamma}\circ(g\circ\ell)\;:\;S^1\rightarrow\mathcal{M}^\prime,
\end{equation*}
where we have used the factorisation Eq.\ \eqref{eqn:factorisation}. Its spectral flow is, by Prop.\ \ref{prop:boundary.independence}, just the winding number of $g\circ\ell$ around the origin. For \emph{any} loop $\ell_{w^+}$ encircling only $w^+$ (but not $w^-$) once in the anticlockwise sense, the corresponding Fredholm loop $\ell_{w^+}^{\rm op}$ will have the same spectral flow, by homotopy invariance. So we may replace $\ell_{w^+}$ by a loop $\ell_{w^+,\epsilon}$ of radius $\epsilon$ centred at $w^+$. Then we compute
\begin{equation*}
{\rm sf}(\ell^{\rm op}_{w^+})=-{\rm Wind}(g\circ\ell_{w^+;\epsilon};0)\equiv -{\rm Index}_{w^+}(g) =-1.
\end{equation*}
Similarly, ${\rm sf}(\ell^{\rm op}_{w^-})=-{\rm Index}_{w^-}(g) =+1$ for \emph{any} operator loop going anticlockwise around $w^-$ (but not $w^+$) once.

In conclusion, we see that spectral flow structure of $\wt{H}^{\rm Weyls}(\gamma)$ is simply pulled back from that of Weyl Hamiltonian via the classifying vector field $g$, and we can determine the Fermi arc topology by inspecting the windings of $g$.

\subsubsection{Connectedness of Fermi arcs}
Strictly speaking, we have only shown that a Fermi arc emanates out of $w^+$ and also $w^-$, but not that they must join up and form a continuous arc from $w^+$ to $w^-$. They could actually separately branch off to infinity (the simplified alternative model of \cite{BGLM} has this feature). 
To exclude this possibility, let us study the winding of $g$ along the infinite line $\ell$ which orthogonally bisects the line segment $l$ joining $w^+$ to $w^-$. This winding makes sense as $g$ points to the right at both ends of this infinite line. Now, an elementary calculation shows that $g(\frac{w^++w^-}{2})$ points to the \emph{left} at the point of intersection $\frac{w^++w^-}{2}$ between $\ell$ and $l$. Furthermore, the vertical component of $g$ at a point on $\ell$ lying on one side of $l$, is opposite to that of the corresponding point on the other side of $l$. This means that the direction of $g$ winds around once fully as $\ell$ is traversed, see Fig.\ \ref{fig:Weyl.vector.field}. Therefore there is spectral flow along $\ell$ (of $+1$ or $-1$ depending on orientation convention), so the Fermi arc must intersect $\ell$ once (up to signs), at some point depending on the boundary condition $\gamma$. We can translate $\ell$ in a manner parallel to $l$, and the same winding calculation holds until $\ell$ hits either $w^+$ or $w^-$. Therefore the overall Fermi arc joins up $w^+$ and $w^-$.

\subsection{Homology of Fermi arcs in continuum models via vector fields}\label{sec:homology.perspective}
{\bf Creation/annihilation of Weyl points.} We can easily construct a homotopy $g_t, t\in[0,1]$ of the vector field such that initially $g_0=g$, the projected Weyl points $w^\pm$ move towards each other as $t$ is increased, and finally collide at their midpoint $w_0=\frac{w^+ + w^-}{2}$ when $t=1$. The zero $w_0$ of the final vector field $g_1$ has zero local index, and can be removed by a further homotopy \cite{Milnor}. The resultant non-singular vector field yields, after composition with $\wt{h}^{\rm Weyl;\gamma}$, the nowhere-vanishing symbol of some differential operator on the half-plane. Then the essential spectrum is gapped everywhere in the contractible space $\widehat{\RR}^2$, and there is no spectral flow at all. This is what is meant by a connected pair of Weyl points being annihilated (or created, by running the homotopy in reverse).

In the language of Euler structures and Euler chains, utilised in the Weyl semimetal context in \cite{MT}, the deformation classes of the singular vector field $g$ (regarded as being defined over compactified momentum space $\widehat{\RR}^2\cup\{\infty\}\simeq S^2$) which preserve singularities $w^+, w^-$ (local creation and annihilation of pairs of singularities are allowed) is called an \emph{Euler structure} \cite{Turaev} relative to $\{w^+,w^-\}$. Obstruction theory says that these Euler structures are classified by \mbox{$H^1(S^2\setminus\{w^+,w^-\})\cong\ZZ$}; Poincar\'{e} dually, the \emph{Euler chains} connecting $w^+$ to $w^-$ are classfied by $H_1(S^2,\{w^+,w^-\})$. The Fermi arc is an Euler chain with end points $w^\pm$, and Poincar\'{e} duality says that its class in $H_1(S^2,\{w^+,w^-\})$ is already specified by counting its intersection numbers with various homology 1-cycles (i.e.\ non-contractible loops) in $S^2\setminus\{w^+,w^-\}$.

\subsubsection{Bulk Dirac string projects homologically to Fermi arc} 
The bulk operator $H^{\rm Weyls}$ is itself given by a 3-component vector field 
\begin{equation*}
\breve{g}:\widehat{\RR}^3\ni \vect{p}\mapsto (g(p_x,p_y),p_z)\in \widehat{\RR}^3,
\end{equation*}
according to its decomposition into a $2\times 2$ matrix family in Eq. \eqref{eqn:Weyls.symbol}. This $\breve{g}$ has zeros precisely at the Weyl points $\vect{p}^\pm=(p_x^\pm,p_y^\pm,p_z^\pm)=(w^\pm,0)$. At each fixed $p_z$, this $\breve{g}(\cdot,\cdot,p_z)=g$ may be regarded as a vector field on $S^2$ (compactification of $\widehat{\RR}^2$) as above, so that $\breve{g}$ gives a vector field on $S^2\times\widehat{\RR}$. Furthermore, on the $p_z\rightarrow\infty$ limit sphere, $\breve{g}$ points in the same direction, and similarly on the $p_z\rightarrow -\infty$ limit sphere. So overall, $\breve{g}$ is a vector field on the \emph{suspension} ${\mathtt S}S^2\cong S^3$, with $\breve{g}$ being non-singular except at the two points $\vect{p}^\pm$. Except for $\vect{p}^\pm$, the $2\times 2$ matrix $H^{\rm Weyls}(\vect{p})$ has a gap around 0, and thus a one-dimensional negative eigenspace, corresponding to the direction $\breve{g}(\vect{p})/|\breve{g}(\vect{p})|\in S^2\cong\CC\PP^1$ (Bloch sphere identification). 

Globally, $H^{\rm Weyls}$ has a negative energy eigenbundle $\mathcal{E}_-$ well-defined over \mbox{$S^3\setminus \{\vect{p}^+,\vect{p}^-\}$}, with first Chern class in $H^2(S^3\setminus \{\vect{p}^+,\vect{p}^-\})$. Restricted to small two-spheres centered at $\vect{p}^\pm$, $\mathcal{E}_-$ has local Chern class given by the $\pm$ local index of $\breve{g}$ at $\vect{p}^\pm$ --- this is the local ``Dirac monopole picture'' of the Weyl points that we saw in Remark \ref{rem:Dirac.string}. 
Poincar\'{e} dually, the Chern class of $\mathcal{E}_-$ becomes a homology class in 
\begin{equation*}
H_1(S^3,\{\vect{p}^+,\vect{p}^-\})\cong\ZZ,
\end{equation*}
represented by a \emph{Dirac string} connecting $\vect{p}^+$ to $\vect{p}^-$. Physically, if one removes a neighbourhood of the Dirac string, then $\mathcal{E}_-$ becomes trivialisable, and can be given a global Berry connection 1-form.

The map $\pi$ projecting out the suspension variable takes $\vect{p}^\pm$ to $w^\pm$. Accordingly, there is a homology projection $\pi_*:H_1(S^3,\{\vect{p}^+,\vect{p}^-\})\rightarrow H_1(S^2,\{w^+,w^-\})$, taking the bulk Dirac string to the Fermi arc --- we saw the tight-binding model version of this observation in \S\ref{sec:homology.Fermi.arc}. In more detail, the suspension quotient $S^2\times [0,1]\rightarrow {\mathtt S}S^2\cong S^3$ induces, by Poincar\'{e} duality, the map $H_1(S^3,\{\vect{p}^+,\vect{p}^-\})\rightarrow H_1(S^2\times [0,1],\{\vect{p}^+,\vect{p}^-\})$, which composes with the natural map $H_1(S^2\times [0,1],\{\vect{p}^+,\vect{p}^-\})\rightarrow H_1(S^2,\{w^+,w^-\})$ to give $\pi_*$ above.

As mentioned in \cite{MT}, the actual geometric Fermi arc in the homology class is as yet undetermined. The missing data is that of the boundary conditions $\gamma$, and possibly a perturbation potential $V$. We showed in Prop. \ref{prop:boundary.independence} and Theorem \ref{thm:perturbation.independence} that these extra data do not change the topology (i.e.\ homology) of the Fermi arc.

\subsection{Producing Weyl points at arbitrary $p_z$}\label{sec:shifted.Weyl.points}
The model Eq.\ \eqref{eqn:Weyls.symbol} with $g$ given by Eq.\ \eqref{eqn:model.vector.field}, has Weyl points $\vect{p}^\pm=(p_x^\pm,p_y^\pm,0)$ projecting onto $w^\pm=(p^\pm_x,p^\pm_y)$ where $w^\pm$ are the zeros of $g$. We wish to make a modification such that $\vect{p}^\pm$ have general non-zero $p_z$-components $p_z^\pm$.

Suppose, without loss, that $p^+_x\neq p^-_x$, and let $a, b$ be the appropriate constants such that $p^\pm_z=ap^\pm_x+b$. This gives a twisting function $A:(p_x,p_y)\mapsto ap_x+b\in\RR$. Take the $A$-twisted symbol
\begin{equation*}
H^{\rm Weyls}_A(\vect{p})=\begin{pmatrix} p_z-(ap_x +b) &  \overline{g(p_x,p_y)} \\ g(p_x,p_y) & -(p_z-(ap_x+b)) \end{pmatrix}.
\end{equation*}
By construction, it vanishes exactly when $g=0$ and $p_z=ap_x+b$, i.e.\ when $\vect{p}=(p_x^\pm,p_y^\pm,p_z^\pm)\equiv\vect{p}^\pm$ as desired.

Fourier transforming $H^{\rm Weyls}_A$ only along $x,y$, produces a family of \emph{twisted} 1D Dirac Hamiltonians,
\begin{equation*}
H^{\rm 1D}_{{\rm Weyls},A}(p_x,p_y)=\begin{pmatrix}-i\frac{d}{dz}-A(p_x,p_y) & \overline{g(p_x,p_y)} \\ g(p_x,p_y) & i\frac{d}{dz}+A(p_x,p_y) \end{pmatrix},\qquad (p_x,p_y)\in\widehat{\RR}^2,
\end{equation*}
where $A\equiv A(p_x,p_y)=ap_x+b$. Similarly, the half-space $\wt{H}^{\rm Weyls}_A(\gamma)$ with boundary conditions $\gamma$ decomposes into twisted half-line Dirac Hamiltonians,
\begin{equation*}
\wt{H}^{\rm Weyls}_A(\gamma)\rightsquigarrow \wt{H}^{\rm 1D}_{A(p_x,p_y)}\left(g(p_x,p_y);\gamma(p_x,p_y)\right),\qquad (p_x,p_y)\in\widehat{\RR}^2.
\end{equation*}
Because $A$ depends linearly on $p_x, p_y$, we can safely gauge transform it away while still preserving the spectral flow structure of the half-line operator family $\{\wt{H}^{\rm 1D}_{{\rm Weyls},A}(p_x,p_y;\gamma)\}_{(p_x,p_y)\in\widehat{\RR}^2}$, as argued in \S\ref{sec:gauge.invariance}, Prop.\ \ref{prop:Wahl.twisted.sf}. This means that only the windings of $g$ around the projected Weyl points $w^\pm$ matter for the Fermi arc, and the analysis of \S\ref{sec:generic.model}-\S\ref{sec:homology.perspective} holds without the requirement that $p_z^\pm=0$.

\section*{Acknowledgements}
The author thanks G.\ De Nittis and K.\ Yamamoto for stimulating discussions, and A.\ Carey for clarifying some technicalities of spectral flow. He also acknowledges support from the Australian Research Council, via research grants DE170900149 and DP200100729.

\appendix
\section{Spectrum of half-line Dirac Hamiltonians}\label{sec:half.line.spectrum}
\subsection{Deficiency indices of half-line Dirac Hamiltonians}\label{sec:deficiency.indices}
It is well-known that (minus) the momentum operator $i\frac{d}{dz}$ on the half-line cannot be made self-adjoint. Physically, the dynamics generated by $i\frac{d}{dz}$ on $L^2(\RR)$ is that of left translation with unit speed, and this is unitary (conserves probability) on the full line Hilbert space $L^2(\RR)$. But once a boundary at $z=0$ is introduced, there are no right-movers to reflect into, and the left-movers keep getting absorbed by the boundary.

To compensate for this problem, one adds as a direct sum, the operator $-i\frac{d}{dz}$, so that the left/right movers exactly compensate. However, another issue arises in the \emph{uniqueness} of this procedure. Namely, a left-moving wave $e^{ip(z+t)}$ could be superposed with a right-moving $e^{ip(z-t)}$ (having the same energy $p$) with some ${\rm U}(1)$ phase shift constituting a choice of boundary condition --- this is another way of understanding the spin polarisation condition in Eq. \eqref{eqn:half.line.Dirac.domain}. 

These issues can be systematically addressed in von Neumann's theory of deficiency indices \cite{RS2}. Roughly speaking, these indices count the number of $\pm i$ eigenvalues that need to be eliminated when constructing self-adjoint extensions that generate genuine unitary dynamics. For example, $z\mapsto e^{-z}$ is a normalisable eigenfunction for $\pm i\frac{d}{dz}$ with imaginary eigenvalue $\mp i$.

We start with the massless case, $\wt{H}^{\rm 1D}(0,0)=-i\frac{d}{dz}\oplus i\frac{d}{dz}$, which is first defined as a symmetric operator on $C_0^\infty(\RR_+)^{\oplus 2}$, the smooth functions with compact support away from the boundary. This is easily seen to have deficiency indices $(1,1)$. Accordingly, there is a ${\rm U}(1)$-family of possible self-adjoint extensions, and they are explicitly given by $\wt{H}^{\rm 1D}(0,0;\gamma)$.

For the massive half-line Dirac Hamiltonians, $\wt{H}^{\rm 1D}(m,\theta)$, they also have deficiency indices $(1,1)$, due to the stability of the latter against the bounded perturbation (the mass term) \cite{Kato}. In fact, deficiency indices are stable against relatively bounded perturbations with bound smaller than $1$, and even more generally, see \cite{BF}. Thus the $\wt{H}^{\rm 1D}(m,\theta;\gamma)$ of Eq.\ \eqref{eqn:half.line.Dirac.domain} give the most general self-adjoint half-line Dirac Hamiltonians. 

\subsection{Essential spectrum of half-line Dirac Hamiltonians}\label{sec:half.line.spectrum2}
Due to the finite deficiency indices, for any boundary condition $\gamma$, and $\lambda\in\CC\setminus\RR$, the resolvents $(\wt{H}^{\rm 1D}(m,\theta;\gamma)-\lambda)^{-1}$ and $(\wt{H}^{\rm 1D}(m,\theta;0)-\lambda)^{-1}$ only differ by a finite-rank, thus compact operator (\cite{RS4} XIII.4, Example 5). Consequently, Weyl's essential spectrum theorem (\cite{RS4} Theorem XIII.14) says that
\begin{equation*}
\sigma_{\rm ess}\left(\wt{H}^{\rm 1D}(m,\theta;\gamma)\right)=\sigma_{\rm ess}\left(\wt{H}^{\rm 1D}(m,\theta;0)\right),\qquad \forall\; e^{i\gamma}\in{\rm U}(1).
\end{equation*}
Together with the unitary equivalence $\wt{H}^{\rm 1D}(m,\theta;\gamma)\cong \wt{H}^{\rm 1D}(m,\theta-\gamma;0)$, it therefore suffices to consider $\gamma=0=\theta$ as far as $\sigma_{\rm ess}$ is concerned. 

It is convenient to carry out a further unitary spin rotation $\frac{1}{\sqrt{2}}\begin{pmatrix} 1 & 1 \\ i & -i\end{pmatrix}$ in the $\CC^2$ degrees of freedom, which effects
\begin{equation*}
\wt{H}^{\rm 1D}(m,0;0)= \begin{pmatrix} -i\frac{d}{dz} & m \\ m & i\frac{d}{dz} \end{pmatrix} \overset{\cong}{\longrightarrow} \begin{pmatrix} m & -\frac{d}{dz} \\ \frac{d}{dz} & -m \end{pmatrix},
\end{equation*}
and transforms the $\gamma=0$ boundary condition to read $\psi(0)\propto \begin{pmatrix} 1 \\ 0\end{pmatrix}$. Notice that the latter entails a Dirichlet condition on the second component of the spinor, and no condition on the first component. Explicitly, the domain of self-adjointness is $\mathring{H}^1(\RR_+)\oplus H^1(\RR_+)$, where $\mathring{H}^1(\RR_+)$ denotes $H^1(\RR_+)$ subject to the vanishing condition at $z=0$. For the first component, we can still use the odd (sine) Fourier transform, so the spectral problem for $\wt{H}^{\rm 1D}(m,0;0)$ may be analysed in momentum space. Similarly to the boundaryless case, we find that the full spectrum is 
\begin{equation*}
\sigma(\wt{H}^{\rm 1D}(m,0;0))=(-\infty,-m]\cup[m,\infty),
\end{equation*}
which therefore completely comprises essential spectrum. Then Eq. \eqref{eqn:Dirac.essential.spectrum} follows.

\subsection{Relatively compact perturbations}\label{sec:perturbations}
Let us abbreviate $\wt{H}^{\rm 1D}(m,0;0)$ to $H_0$. The resolvent $(H_0-\lambda)^{-1}$ at $\lambda\in\CC\setminus\sigma(H_0)$ can be computed explicitly to have the integral kernel (see pp. 4 of \cite{Enblom})
\begin{align*}
K(z,z^\prime;\lambda)&=\frac{1}{2}\begin{pmatrix}-\frac{\lambda+m}{i\mu} & {\rm sgn}(z-z^\prime) \\ -{\rm sgn}(z-z^\prime) & -\frac{\lambda-m}{i\mu}\end{pmatrix}e^{i\mu |z-z^\prime|}\\ &\qquad -\frac{1}{2}\begin{pmatrix}\frac{\lambda+m}{i\mu} & 1\\ 1 & \frac{\lambda-m}{i\mu}\end{pmatrix}e^{i\mu (z-z^\prime)},\qquad z,z^\prime\in\RR_+,
\end{align*}
where $\mu=(\lambda^2-m^2)^{1/2}$ has ${\rm Im}\,\mu >0$. This kernel is square-integrable on bounded subsets, so if $V=V(z)$ is a $2\times 2$ Hermitian matrix-valued continuous potential on $\RR_+$ that vanishes at infinity, we can, by truncating kernels, approximate $V(H_0-\lambda)^{-1}$ in norm using Hilbert--Schmidt operators; thus $V(H_0-\lambda)^{-1}$ is compact. This argument to establish the \emph{$H_0$-relatively compactness} of $V$ is of a fairly standard form, and works for potentials which are merely $L^2$ up to a $\epsilon$-essentially bounded piece, see e.g.\ \cite{RS4} XIII.4 Example 6.

It follows that the resolvent difference
\begin{align*}
(H_0+V-\lambda)^{-1}-(H_0-\lambda)^{-1}&=(H_0+V-\lambda)^{-1}(1-(H_0-\lambda+V)(H_0-\lambda)^{-1})\\
&=-(H_0+V-\lambda)^{-1}V(H_0-\lambda)^{-1}
\end{align*}
is compact. Then Weyl's essential spectrum theorem applies, saying that $H_0$ and $H_0+V$ have the same essential spectrum, i.e.,
\begin{equation*}
\sigma_{\rm ess}(\wt{H}^{\rm 1D}(m,0;0)+V)=\sigma_{\rm ess}(\wt{H}^{\rm 1D}(m,0;0)).
\end{equation*}

Finally, since $\wt{H}^{\rm 1D}(m,\theta;\gamma)$ and $\wt{H}^{\rm 1D}(m,0;\gamma+\theta)$ are unitarily equivalent, and the resolvents of $\wt{H}^{\rm 1D}(m,0;\gamma+\theta)$ and $\wt{H}^{\rm 1D}(m,0;0)$ only differ by a compact operator (\S\ref{sec:half.line.spectrum2}), we also see that $V$ is also compact relative to any $\wt{H}^{\rm 1D}(m,\theta;\gamma)$. Thus for the perturbed half-line Dirac Hamiltonians $\wt{H}^{\rm 1D}_V(m,\theta;\gamma)\equiv \wt{H}^{\rm 1D}(m,\theta;\gamma)$,
\begin{equation*}
\sigma_{\rm ess}(\wt{H}^{\rm 1D}_V(m,\theta;\gamma))=\sigma_{\rm ess}(\wt{H}^{\rm 1D}(m,\theta;\gamma)),\qquad \forall\; (m,\theta,\gamma)\in\widehat{\RR}^2\times{\rm U}(1)\cong\mathcal{M}.
\end{equation*}


\end{document}